\documentclass[journal, twoside]{IEEEtran}

\usepackage{graphicx}
\usepackage{subcaption}
\usepackage{tikz,xparse}

\newcommand\lineWidth{0.65pt}
\newcommand\markSize{2.25}
\usetikzlibrary{patterns,dsp,chains,matrix,calc,backgrounds,circuits.logic.US,mux,fit}
\pgfdeclarelayer{background}
\pgfdeclarelayer{foreground}
\pgfsetlayers{background,main,foreground}

\usepackage{mathptmx}
\usepackage{verbatim}
\usepackage{calc}%
\usepackage{ifthen}
\usepackage{xifthen}
\usepackage{cancel}
\usepackage{bm}
\usepackage{verbatim}
\usepackage{multirow}
\usepackage{cite}
\usepackage[hyphens]{url}
\usepackage[nolist]{acronym} 
\usepackage{pgfplots}
\pgfplotsset{
/pgfplots/short line/.style={
    legend image code/.code={
        \draw[mark repeat=2,mark phase=2,#1] 
            plot coordinates {(0cm,0cm) (0.1275cm,0cm) (0.245cm,0cm)};
    },
},
}
\pgfdeclareplotmark{hexagon}
{
\newdimen\R
\R=2.25pt
\draw (0:\R) \foreach \x in {60,120,...,300} {
        -- (\x:\R)
    }-- cycle (90:\R);
}

\usetikzlibrary{arrows,shapes,graphs,graphs.standard,quotes,arrows.meta,decorations.markings,positioning}
\usetikzlibrary{fit,positioning,hobby}
\usetikzlibrary{decorations.text}
\usetikzlibrary{decorations.pathreplacing}
\usepackage{transparent}

\pgfplotsset{compat=newest}
\usepackage[bookmarks=false]{hyperref}

\usepackage{nicefrac}
\usepackage{siunitx}
\usepackage{amsmath, amsbsy, amssymb, latexsym}
\usepackage{dsfont}
\usepackage{amsthm}
\usepackage{cuted}
\usepackage{mathtools}

\usepackage{nicematrix}
\usepackage{booktabs}

\hypersetup{bookmarksdepth=-2}
\usepackage{comment}
\usepackage[utf8]{inputenc}
\usepackage{xcolor}
\usepackage{enumitem}
\usepackage[linesnumbered,titlenumbered,ruled,vlined,resetcount]{algorithm2e}
\usepackage{color}

\SetNlSty{textbf}{\color{black}}{}

\SetCommentSty{mycommentfont}

\SetKwRepeat{Do}{do}{while}
\SetKwProg{Fn}{subroutine}{:}{end}

\SetKwIF{NewIf}{NewElseIf}{NewElse}{if~(\endgraf}{\endgraf)~then}{else if}{else}{end if}%

\let\oldnl\nl%
\newcommand{\nonl}{\renewcommand{\nl}{\let\nl\oldnl}}%

\usepackage{setspace}
\usepackage{algpseudocode}

\makeatletter
\patchcmd\algocf@Vline{\vrule}{\vrule \kern-0.4pt}{}{}
\patchcmd\algocf@Vsline{\vrule}{\vrule \kern-0.4pt}{}{}
\makeatother

\usepackage[normalem]{ulem}
\usepackage{newtxmath}

\usepackage{marginnote}
\tikzset{>=latex}
\usepackage{floatrow}
\SetAlCapNameFnt{\footnotesize}
\SetAlCapFnt{\footnotesize}

\DeclareMathAlphabet{\mathcal}{OMS}{cmsy}{m}{n}
\DeclareMathAlphabet{\mathbb}{U}{msb}{m}{n}

\DeclareMathOperator{\w}{\textit{w}_\mathrm{H}}
\DeclareMathOperator{\wmin}{\textit{w}_\mathrm{min}}
\DeclareMathOperator{\Awmin}{\textit{A}_{\textit{w}_\mathrm{min}}}

\DeclareMathOperator{\dmin}{\textit{d}_\mathrm{min}}
\DeclareMathOperator{\Admin}{\textit{A}_{\textit{d}_\mathrm{min}}}
\DeclareMathOperator{\C}{\mathcal{P}(\mathcal{I}, \boldsymbol{T})}
\DeclareMathOperator{\Ci}{\mathcal{P}_\textit{i}(\mathcal{I}, \boldsymbol{T})}

\DeclareMathOperator{\RM}{\mathcal{RM}}

\newcommand{\ZZ}{\mathbb{Z}}

\newcommand{\LB}{\left(}
\newcommand{\RB}{\right)}

\newcommand{\ie}{i.\,e.,\ }
\newcommand{\eg}{e.\,g.,\ }

\newcommand\mmsq[1]{\SI{#1}{\milli\meter\squared}}

\newcommand\dB[1]{\SI{#1}{\decibel}}
\newcommand\nm[1]{\SI{#1}{\nano\meter}}

\definecolor{mittelblau}{RGB}{0, 126, 198}
\definecolor{violettblau}{cmyk}{0.9, 0.6, 0, 0}
\definecolor{rot}{RGB}{238, 28 35}
\definecolor{apfelgruen}{RGB}{140, 198, 62}
\definecolor{gelb}{RGB}{255, 229, 0}
\definecolor{orange}{RGB}{244, 111, 33}
\definecolor{pink}{RGB}{237, 0, 140}
\definecolor{lila}{RGB}{128, 10, 145}
\definecolor{hellgrau}{RGB}{224, 224, 224}
\definecolor{mittelgrau}{RGB}{128, 128, 128}
\definecolor{dunkelgrau}{RGB}{80,80,80}
\definecolor{anthrazit}{RGB}{19, 31, 31}
\definecolor{darkgreen}{RGB}{34,139,34}
\definecolor{aqua}{RGB}{0, 255, 255}

\definecolor{turquoise}{RGB}{3, 148, 156}
\definecolor{neuesgruen}{RGB}{61, 173, 65}
\definecolor{dunklereshellgrau}{RGB}{176, 176, 176}
\definecolor{lightgray}{RGB}{211,211,211}
\definecolor{neuesgelb}{RGB}{255,160,0}
\definecolor{neuescyan}{RGB}{69,185,224}
\definecolor{neuesmagenta}{RGB}{197,67,143}
\definecolor{grape}{HTML}{6945C2}
\definecolor{tyrianpurple}{HTML}{6D7BFF}
\definecolor{navyblue}{HTML}{172370}
\definecolor{bluegrotto}{HTML}{009277}
\definecolor{redB}{HTML}{CE054A}

\definecolor{rptured}   {HTML}{e31b4c}
\definecolor{rptublue}  {HTML}{042c58}
\definecolor{rptulblue} {HTML}{6ab2e7}
\definecolor{rptugreen} {HTML}{26d07c}
\definecolor{rptuorange}{HTML}{ffa252}

\begin{acronym}
\acro{ML}{maximum likelihood}
\acro{BP}{belief propagation}
\acro{BPL}{belief propagation list}
\acro{LDPC}{low-density parity-check}
\acro{BER}{bit error rate}
\acro{SNR}{signal-to-noise-ratio}
\acro{BPSK}{binary phase shift keying}
\acro{AWGN}{additive white Gaussian noise}
\acro{LLR}{Log-Likelihood Ratio}
\acro{MAP}{maximum a posteriori}
\acro{FER}{frame error rate}
\acro{BLER}{block error rate}
\acro{SCL}{successive cancellation list}
\acro{SC}{successive cancellation}
\acro{BI-DMC}{Binary Input Discrete Memoryless Channel}
\acro{CRC}{cyclic redundancy check}
\acro{CA-SCL}{CRC-aided successive cancellation list}
\acro{PAC}{polarization-adjusted convolutional}
\acro{BEC}{Binary Erasure Channel}
\acro{BSC}{Binary Symmetric Channel}
\acro{BCH}{Bose-Chaudhuri-Hocquenghem}
\acro{RM}{Reed--Muller}
\acro{RS}{Reed-Solomon}
\acro{SISO}{soft-in/soft-out}
\acro{3GPP}{3rd Generation Partnership Project }
\acro{eMBB}{enhanced Mobile Broadband}
\acro{CN}{check node}
\acro{VN}{variable node}
\acro{GenAlg}{Genetic Algorithm}
\acro{CSI}{Channel State Information}
\acro{OSD}{ordered statistic decoding}
\acro{MWPC-BP}{minimum-weight parity-check BP}
\acro{FFG}{Forney-style factor graph}
\acro{MBBP}{multiple-bases belief propagation}
\acro{URLLC}{ultra-reliable low-latency communications}
\acro{DMC}{discrete memoryless channel}
\acro{SGD}{stochastic gradient descent}
\acro{QC}{quasi-cyclic}
\acro{5G}{fifth generation mobile telecommunication}
\acro{SCAN}{soft cancellation}
\acro{LSB}{least significant bit}
\acro{MSB}{most significant bit}
\acro{AED}{automorphism ensemble decoding}
\acro{AE-SC}{automorphism ensemble successive cancellation}
\acro{PPV}{Polyanskyi-Poor-Verd\'{u}}
\acro{RREF}{reduced row echelon form}
\acro{PTPC}{pre-transformed polar code}
\acro{SCAL}{successive cancellation automorphism list}
\acro{DE}{density evolution}
\acro{PPV}{Polyanskyi-Poor-Verd\'{u}}
\acro{DBT}{dichotomous binary tree}
\acro{PDBT}{perfect dichotomous binary tree}
 
\acro{IBE}{information bit extraction}
\acro{DR}{delta recovery}
\acro{FSSCL}[Fast-SSCL]{fast simplified successive cancellation list}
\acro{HD}{hard decision}
\acro{PAR}{placement and routing}
\acro{PFT}{polar factor tree}
\acro{PM}{path metric}
\acro{PVT}{process, voltage and temperature}
\acro{REP}{repetition}
\acro{SPC}{single parity check}
\end{acronym}
\input{figures/register.tex}

\newtheorem{theorem}{Theorem}

\newtheorem{remark}{Remark}
\newtheorem{example}{Example}
\AtBeginEnvironment{example}{%
  \pushQED{\qed}%
}
\AtEndEnvironment{example}{\popQED\endexample}

\newcommand\toremove[1]{} %

\newcommand\ourTitle{Row-Merged Polar Codes:\\Analysis, Design and Decoder Implementation}

\newcommand\todo[1]{\textcolor{orange}{#1}}

\newcommand\comaz[1]{{\textcolor{mittelblau}{#1}}}

\begin{document}

\begin{NoHyper}
\title{\ourTitle}

\author{%
    Andreas~Zunker,
    Marvin~Geiselhart~\IEEEmembership{Graduate Student Member,~IEEE,}\\
    Lucas~Johannsen,~\IEEEmembership{Graduate Student Member,~IEEE,}
    Claus~Kestel,~\IEEEmembership{Graduate Student Member,~IEEE,}\\
    Stephan~ten~Brink,~\IEEEmembership{Fellow,~IEEE,}
    Timo~Vogt, and
    Norbert~Wehn,~\IEEEmembership{Senior Member,~IEEE,}%
    \thanks{
        Andreas Zunker, Marvin Geiselhart and Stephan ten Brink are with the
        Institute of Telecommunications, 
        University of Stuttgart, 70569 Stuttgart, Germany
        (e-mail:  zunker@inue.uni-stuttgart.de; geiselhart@inue.uni-stuttgart.de; 
        tenbrink@inue.uni-stuttgart.de)
    }%
    \thanks{
        Lucas Johannsen, Claus Kestel and Norbert Wehn are with the
        Microelectronic Systems Design Research Group,
        RPTU Kaiserslautern-Landau, 67663 Kaiserslautern, Germany 
        (e-mail: lucas.johannsen@rptu.de; kestel@rptu.de; 
        norbert.wehn@rptu.de)
    }%
    \thanks{
        Lucas Johannsen and Timo Vogt are with the 
        Faculty of Engineering, 
        Koblenz University of Applied Sciences, 56075 Koblenz, Germany
        (e-mail: johannsen@hs-koblenz.de; vogt@hs-koblenz.de)
    }%
    \thanks{
    This work is supported by the German Federal Ministry of Education and Research (BMBF) within the project Open6GHub (grant no. 16KISK019 and 16KISK004).
    }
}

\maketitle

\begin{abstract}
    Row-merged polar codes are a family of \acp{PTPC} with little precoding overhead.
Providing an improved distance spectrum over plain polar codes, they are capable to perform close to the finite-length capacity bounds.
However, there is still a lack of efficient design procedures for row-merged polar codes.
Using novel weight enumeration algorithms with low computational complexity,
we propose a design methodology for row-merged polar codes that directly considers their minimum distance properties.
The codes significantly outperform state-of-the-art \ac{CRC}-aided polar codes under \ac{SCL} decoding in error-correction performance. 
Furthermore, we present \ac{FSSCL} decoding of \acp{PTPC}, based on which we derive a high-throughput, unrolled architecture template for fully pipelined decoders.
Implementation results of \ac{SCL} decoders for row-merged polar codes in a 12\,nm technology additionally demonstrate the superiority of these codes with respect to the implementation costs, compared to state-of-the-art reference decoder implementations.

\end{abstract}
\acresetall


\section{Introduction}

\IEEEPARstart{P}{olar} codes, originally introduced by Arıkan, have been shown to asymptotically achieve channel capacity \cite{ArikanMain}. 
However, their performance degrades significantly in the short and medium block length regimes. 
This degradation can be attributed to two main factors: first, the use of suboptimal \ac{SC} decoding, and second, a poor distance spectrum. 
To improve decoding performance, list decoding has been proposed \cite{talvardyList}, which can approach the \ac{ML} performance of the code for a sufficiently large list size. 
Since then, designs specifically tailored for \ac{SCL} decoding have been suggested, taking into account the minimum distance properties \cite{RMurbankePolar}, \cite{HybridTse}. 
Ultimately, the performance of plain polar codes and the related \ac{RM} codes remains limited by their achievable distance spectrum.
In particular, polar codes have a poor minimum distance $\dmin$, whereas \ac{RM} codes have a usually sufficient $\dmin$, but a large number of minimum-weight codewords $\Admin$.

A key advancement in improving the distance spectrum by removing low-weight codewords is the introduction of precoding. 
The resulting code is a subcode of the polar code \cite{subCodes} and will be referred to as \ac{PTPC} hereafter.
Proper design of the pre-transform results in very powerful coding schemes that can be decoded by \ac{SCL} with little complexity overhead. 
Several different precoders have been proposed in the past:
In \cite{talvardyList} and \cite{NuiChen2012crc}, \ac{CRC} codes have been suggested as pre-transforms to prevent incorrect decisions of the \ac{SCL} decoder. %
Similarly, the authors of \cite{DynamicFrozenBits} and \cite{ParityCheckPolarCode} propose parity checks, which can be seen as dynamic frozen bits.
Outer convolutional codes have been proposed in \cite{arikan2019pac}, resulting in so-called \ac{PAC} codes, which can be also decoded using the list decoder \cite{Rowshan2021listpac}.
Furthermore, concatenated  \cite{BP_flexible} and multiple concatenated (``deep'') polar codes \cite{choi2023deeppolar} have been proposed.
Very recently, it has been shown that also repetition codes are effective in eliminating minimum-weight codewords with very low implementation complexity; the resulting code is a row-merged polar code \cite{GelincikRowMerged}. Note that information bit repetition was previously introduced in \cite{RowshanRepetition2019}, however, to be used in iterative decoding rather than to directly assist a single-pass \ac{SCL} decoder in the form of dynamic frozen bits.

Besides the code properties including the error-correction performance, the
hardware implementation costs of respective decoders are important for
practical applications. To mitigate the drawback of low throughput and high
latency of \ac{SC}-based decoders, caused by the serial nature of the decoding
algorithm, several algorithmic and architectural optimizations were proposed in
the past. With respect to \ac{SCL} decoding \cite{talvardyList}, the
introduction of \ac{LLR}-based \cite{balatsoukas2015LLRSCL} and \ac{FSSCL}
decoding \cite{sarkis2016FastList, hashemi2017flexfastsscl} directly improve the implementation metrics of the corresponding decoder
hardware on the algorithmic
level. Unrolled and pipelined decoder architectures achieve high throughput
and are very efficient since they are data flow dominated and minimize the
control flow \cite{schlaefer2013unrolledLDPC, giard2016unrolledList,
kestel2020unrolledList}. Optimized candidate generation in substituent nodes of
the decoding tree further reduces the implementation costs
\cite{johannsen2022rate1, johannsen2022spc}.

Our main contributions are summarized as follows:
\begin{itemize}
    \item We propose a novel greedy design algorithm for row-merged polar codes that optimizes the number of minimum-weight codewords using fast weight enumeration methods proposed in \cite{ZunkerTreeIntersection} and \cite{PartialEnumPAC}.
    The designed codes have a smaller number of minimum-weight codewords than the ones in \cite{GelincikRowMerged}, where a heuristic is used instead.
    \item We propose a \ac{FSSCL} decoding algorithm for \acp{PTPC} and its fully unrolled and pipelined hardware implementation.
    \item We give partial weight spectra and \ac{BLER} simulation results for the designed codes %
    and ASIC implementation results of corresponding unrolled decoders.
\end{itemize}
The rest of the paper is structured as follows: 
Sec.~\ref{sec:preliminaries} briefly introduces the notation and the concept of \acp{PTPC} .
In Sec.~\ref{sec:design} we describe our design algorithms for row-merged polar codes.
Sec.~\ref{sec:fastSSCLdf} introduces the \ac{FSSCL} decoding algorithm for \acp{PTPC} and Sec.~\ref{sec:hardware} the corresponding hardware implementation.
In Sec.~\ref{sec:results} we show error-rate performance and implementation results of the proposed codes and decoders, before Sec.~\ref{sec:conc} concludes the paper.

\section{Preliminaries}\label{sec:preliminaries}

\subsection{Notations}
Vectors and matrices are indicated by lowercase and uppercase boldface letters, respectively. 
Sets are sorted and denoted by calligraphic letters.
$\mathbb{Z}_{i:j}$ denotes the set of integers between $i$ and $j-1$, with the shorthand notation $\mathbb{Z}_{j}=\mathbb{Z}_{0:j}$.
An integer $i \in \mathbb Z_{2^n}$ is implicitly represented by its $n$-bit \ac{MSB}-first binary expansion ${i = \sum_{l=0}^{n-1} i_{(l)}\cdot 2^l \mathrel{\widehat{=}} i_{(n-1)} \dots i_{(1)} i_{(0)}}$, with Hamming weight $\w(i)$.

\subsection{Error Probability of Binary Linear Block Codes}
Let $\w(\boldsymbol c)$ denote the Hamming weight of a binary vector $\boldsymbol c \in \mathbb{F}_2^N$.
The minimum distance a of binary linear block code $\mathcal C(N,K)$ with code length $N$ and code dimension $K$ is
\begin{equation*}
    \dmin(\mathcal{C}(N,K)) = \min\limits_{\boldsymbol{c},\boldsymbol{c}' \in \mathcal {C},\, \boldsymbol{c} \neq \boldsymbol {c}'} \w(\boldsymbol{c} \oplus \boldsymbol{c}') = \min\limits_{\boldsymbol {c} \in \mathcal{C},\, \boldsymbol{c} \neq \boldsymbol{0}} \w(\boldsymbol{c}).
\end{equation*}
The number of codewords $\boldsymbol{c} \in \mathcal{C}(N,K)$ with Hamming weight $w$ is given by
\begin{equation*}
    A_w\left(\mathcal{C}(N,K)\right) = \left| \left\{\boldsymbol{c} \in \mathcal{C}(N,K) \,\middle|\, \w(\boldsymbol{c}) = w \right\} \right|  %
\end{equation*} 
and the ordered set $\mathcal{A}_{\mathcal{C}} = \left\{A_0,\,A_1,\,\dots,\,A_N\right\}$ is called the weight spectrum of code $\mathcal{C}(N,K)$. 
For \ac{BPSK} mapping and transmission over an \ac{AWGN} channel with \ac{SNR} $E_\mathrm{b}/N_0$, the \ac{FER} under \ac{ML} decoding $P^\mathrm{ML}_\mathrm{FE}$ of a binary linear block code $\mathcal C(N,K)$ with code rate $R = K/N$ can be upper bounded by the union bound \cite{UnionBound} as 
\begin{equation} %
    P^\mathrm{ML}_\mathrm{FE} \leq P^\mathrm{UB}_\mathrm{FE} = \sum\limits_w A_w \cdot Q\left(\sqrt{2 w \cdot R \cdot \frac{E_\mathrm{b}}{N_0}}\right), \label{eq:unionbound}
\end{equation}
where $Q(\cdot)$ is the complementary cumulative distribution function of the standard normal distribution.
As the \ac{SNR} increases, the contribution of codewords with high weight $w$ in the computation of $P^\mathrm{UB}_\mathrm{FE}$ vanishes. 
Hence, Eq.~(\ref{eq:unionbound}) is well approximated by the first few non-zero terms (\emph{truncated union bound}), and thus, the minimum distance $\dmin$ and the number of minimum-weight codewords $A_{\dmin}$ are two decisive properties for the achievable \ac{ML} performance of a code.

\subsection{Polar Codes}
Polar codes are constructed based on the $n$-fold application of the polar transform, resulting in polarized synthetic bit channels. Reliable bit channels, denoted by the information set $\mathcal{I}$, are used for information transmission, while unreliable bit channels (frozen set $\mathcal{F}=\ZZ_N \setminus \mathcal{I}$) are set to 0. The choice of $\mathcal{I}$ is called the polar code design.
The generator matrix $\boldsymbol G = (\boldsymbol G_{N})_{\mathcal I}$ of an $(N,K)$ polar code  $\mathcal P(\mathcal I)$ is formed by those $K$ rows of the Hadamard matrix $\boldsymbol G_N = \begin{bsmallmatrix} 1 & 0 \\ 1 & 1 \end{bsmallmatrix}^{\otimes n}$ indicated by $\mathcal{I}$, where $(\cdot)^{\otimes n}$ denotes the $n$-th Kronecker product and $N = 2^n$.
Let $\boldsymbol m = \boldsymbol u_\mathcal{I} \in \mathbb F_2^{K}$ denote the message and  $\boldsymbol u_\mathcal{F} = \boldsymbol 0$ the frozen bits. Then, polar encoding can be written as
\begin{equation}\label{eq:polarEncoding}
    \boldsymbol c = \boldsymbol m \cdot \boldsymbol G = \boldsymbol u \cdot \boldsymbol G_N.
\end{equation}

Even though the optimal polar code design $\mathcal I$ is dependent on the transmission channel, some synthetic bit channels are always more reliable than others \cite{bardet_polar_automorphism}. 
Thus, the bit channels exhibit a \textit{partial order}, in which $i \preccurlyeq j$ denotes that channel $j$ is at least as reliable as channel $i$, defined by
\begin{enumerate}
  \item \textit{Left swap}: If $i_{(l)} > j_{(l)}$, and $i_{(l+1)} < j_{(l+1)}$ for ${l\in \mathbb Z_{n-1}}$, while all other $i_{(l')} = j_{(l')}$, $l' \notin \{l,l+1\}$, then $i \preccurlyeq j$. %
    \item \textit{Binary domination}: If $i_{(l)}\le j_{(l)}$ for all $l\in \ZZ_n$, then $i \preccurlyeq j$. 
\end{enumerate}
All remaining relations are derived from transitivity.
In summary, $i \preccurlyeq j$ if and only if $j$ equals $i$ with additional ``1'' bits and/or ``1'' bits moved to more significant positions.
A polar code $\mathcal P(\mathcal I)$ is called a \textit{decreasing monomial code} and said to comply with the partial order ``$\preccurlyeq$'' if and only if
$\mathcal I$ fulfills the partial order as $\forall i \in \mathcal I,\;\forall j \in \mathbb Z_N \text{ with } i \preccurlyeq j \implies j \in \mathcal I$.
In other words, for any information bit $i \in \mathcal I$, all more reliable channels $i \preccurlyeq j < N$ are also information bits and thus $j \in \mathcal I$. 
A decreasing monomial code can be concisely described by its minimal information set $\mathcal I_\mathrm{min}$ \cite{polar_aed}, such that the complete information set can be inferred from the partial order as
\begin{equation*}
    \mathcal I = \bigcup\limits_{i \in \mathcal I_\mathrm{min}} \left\{j \in \mathbb Z_N \,\middle|\, i \preccurlyeq j \right\}.
\end{equation*}
\ac{RM} codes are a special case of decreasing monomial codes.
The \ac{RM} code of order $r$ and length $N=2^n$, denoted by $\RM(r,n)$,
is defined by the minimal information set ${\mathcal I_\mathrm{min} = \{2^{n-r}-1\}}$ and has minimum distance ${\dmin = 2^{n-r}}$.

\subsection{Pre-Transformed Polar Codes (PTPCs)}
\begin{figure}[htp]
	\centering
	\resizebox{.8\columnwidth}{!}{\colorlet{dF}{orange!70!white}%
\begin{tikzpicture}[
    start chain = 1 going below,    node distance = 0pt,
    cell/.style={draw, fill=white, minimum width=1em, minimum height=1em, inner sep=0pt,
                outer sep=0pt, on chain},
    long/.style={draw, fill=white, minimum width=1em, minimum height=2em,  inner sep=0pt,
                outer sep=0pt, on chain},
    tick/.pic = {
    		\draw[line width=0.5pt] (-0.4mm,-0.8mm) -- (0.4mm,0.8mm);
    	}
  ]
  
\node[] (m) at (0.5, 0) {$\boldsymbol m$};
\node[dspsquare] (rp) at (2, 0) {$\mathcal I$};
\node[align=center, anchor=north, above = 0.1cm of rp]  {\footnotesize Rate-\\[-3pt]\footnotesize Profile};
\node[dspsquare] (pt) at (4, 0) {$\boldsymbol T$};
\node[align=center, anchor=north, above = 0.1cm of pt]  {\footnotesize Pre-\\[-3pt]\footnotesize Transform};
\node[dspsquare] (gn) at (6, 0) {$\boldsymbol G_N$};
\node[align=center, anchor=north, above = 0.1cm of gn]  {\footnotesize Polar\\[-3pt]\footnotesize Transform};

\node[] (c) at (7.5, 0) {$\boldsymbol c$};

\draw[dspconn](m)  -- node[] (tickM) {} node[below] {\footnotesize $K$} (rp);
\draw[dspconn](rp) -- node[above,yshift=2pt] {$\boldsymbol v$} node[] (tickV) {} node[below] {\footnotesize $N$}(pt);
\draw[dspconn](pt) -- node[above,yshift=2pt] {$\boldsymbol u$} node[] (tickU) {} node[below] {\footnotesize $N$}(gn);
\draw[dspconn](gn) -- node[] (tickC) {} node[below] {\footnotesize $N$} (c);

\pic at (tickM) {tick};
\pic at (tickV) {tick};
\pic at (tickU) {tick};
\pic at (tickC) {tick};

\end{tikzpicture}}
	\caption{Block diagram of pre-transformed polar encoding. %
    }
	\label{fig:Encoder}
\end{figure}
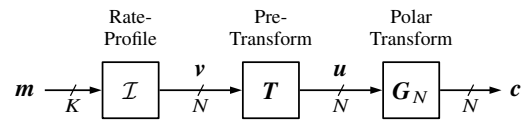

A polar code is said to be \emph{pre-transformed}, if the message is transformed by a (linear) mapping prior to the polar transformation $\boldsymbol G_N$. Pre-transformations can be used to improve the finite-length performance of a polar code $\mathcal{P}(\mathcal{I})$, which may be viewed in two ways: Firstly, in the short block length regime, some frozen channels may not be fully polarized yet and thus, their residual capacity still can be exploited by making them dynamic frozen bits. Secondly, the precoder may increase the minimum distance and/or reduce the number of minimum-weight codewords, thus, improving the \ac{ML} performance of the code. 
Fig.~\ref{fig:Encoder} shows the block diagram of the encoder of a \ac{PTPC} $\mathcal P(\mathcal I,\boldsymbol T) $.
The pre-transformation is sandwiched between the insertion of the frozen bits (rate-profiling) and the polar transform \cite{arikan2019pac}. Hence, it can be described by a multiplication with an upper-triangular matrix $\boldsymbol T \in \mathbb F_2^{N \times N}$, \ie ${t_{i,j}=0}$ if $i > j$ and ${t_{i,i}=1}$ for $i \in \mathcal I$.
The resulting pre-transformed message is given as
\begin{equation}
    \boldsymbol{u} = \boldsymbol{m} \cdot \boldsymbol{T}_\mathcal{I}  = \boldsymbol v \cdot \boldsymbol T,\label{eq:precoder}
\end{equation}
where the rate-profiled message $\boldsymbol{v}$ is found as $\boldsymbol{v}_{\mathcal{I}} = \boldsymbol{m} \in \mathbb{F}_2^K$, $\boldsymbol{v}_\mathcal{F} = \boldsymbol{0}$,
and again $\boldsymbol c = \boldsymbol u \cdot \boldsymbol G_N$.
The pre-transformation introduces dynamic frozen bits whose values are dependent on the previous message bits as
\begin{equation}\label{eq:dynamicFrozenBit}
    u_f = \bigoplus\limits_{h=0}^{f-1} v_h \odot t_{h,f},\;f \in \mathcal F.
\end{equation}
The frozen set can subsequently be split into the set of dynamic frozen bits denoted by ${\mathcal{D}=\left\{d \in \mathcal{F}\,\middle|\, \exists i \in \mathcal{I},\,t_{i,d}=1\right\}}$  
and static frozen bits  $\mathcal{F} \setminus \mathcal{D}$ which are still fixed to 0.
Note that this perspective on \acp{PTPC} differs from conventional code concatenation, where both inner and outer code have a rate $R < 1$. In contrast, \acp{PTPC} as considered in this paper add redundancy solely using the rate profile $\mathcal{I}$, while the following transformations $\boldsymbol{T}$ and $\boldsymbol{G}_N$ are rate-1. While for plain polar codes, simply $\boldsymbol T = \boldsymbol I_{N}$, many different variants of polar codes can be described under the umbrella of \acp{PTPC}\footnote{Any binary linear code can be described as an equivalent \ac{PTPC} \cite{subCodes}.}:

\subsubsection{CRC-aided Polar Codes}
In the formalism of Eq.~(\ref{eq:precoder}), the last $q$ synthetic channels corresponding to the check symbols of the $q$-bit \ac{CRC} of a \ac{CRC}-aided polar code are frozen and turned into dynamic frozen bits using a systematic pre-transformation matrix $\boldsymbol T$. 

\subsubsection{Polarization-Adjusted Convolutional (PAC) Codes}
The pre-transformation is given by a convolution with a generator polynomial 
$p(x) = \sum_{h=0}^q p_h x^h$
with degree $\operatorname{deg}(p(x))=q$ \cite{arikan2019pac}.

\subsubsection{Row-Merged Polar Codes}\label{subsubsec:rowmerged}
Information bits $v_i$ with $i \in \mathcal I$ are repeated in (dynamic) frozen positions $d \in \mathcal{D} \subseteq \mathcal F$, \ie $u_i = u_d = v_i$ \cite{GelincikRowMerged}. 
The set of tuples $(i,d)$ is denoted by $\mathcal{R}$.
Hence, the pre-transformation matrix is $\boldsymbol{T}$, 
with $t_{i,i}=1$ for all $i \in \mathcal{I}$, $t_{i,d} = 1$ for all ${(i,d) \in \mathcal{R}}$, and $t_{i,f}=0$ otherwise.
Note that there can be multiple pairs as $(i,d),(i,d') \in \mathcal{R}$ with $d \neq d'$.
In the following, we denote a row-merged polar code
with rate-profile $\mathcal{I}$ and set of row-merge pairs $\mathcal{R}$ by $\mathcal{P}(\mathcal{I},\mathcal{R})$.

\section{Design of Row-Merged Polar Codes}\label{sec:design}
\subsection{Theoretical Framework}
Our goal is to improve the weight spectrum of polar codes while introducing a minimal overhead caused by the pre-transformation.
To effectively precode polar codes, we first consider their underlying structure.
To simplify the analysis, a \ac{PTPC} $\C = \{\boldsymbol{0}\} \cup \bigcup_{i \in \mathcal I} \Ci$ is partitioned into  $K=|\mathcal I|$ distinct \emph{cosets}
\begin{equation}%
    \Ci = \left\{\boldsymbol{t}_i \cdot \boldsymbol{G}_N \oplus \bigoplus_{h \in \mathcal{H}} \boldsymbol{t}_h \cdot \boldsymbol{G}_N\,\middle|\,\mathcal{H} \subseteq \mathcal{I} \setminus \mathbb{Z}_{i+1} \right\},
\end{equation}
where $\boldsymbol{t}_k \cdot \boldsymbol{G}_N = \bigoplus_{h \in \operatorname{supp}(\boldsymbol{t}_k)} \boldsymbol{g}_h$ with $k \in \mathcal{I} \setminus \mathbb{Z}_i$.
The matrix $\boldsymbol{T}$ is upper-triangular, i.e., $\operatorname{supp}\left(\boldsymbol{t}_h\right) \cap \mathbb{Z}_h = \emptyset$ for all $h \in \mathbb{Z}_N$.
The row $\boldsymbol{g}_i$ is called the \emph{coset leader} of $\Ci$.
Instead of examining the whole code at once, each coset can be considered individually.
We subsequently denote \emph{plain} (i.e., without pre-transform) and \emph{universal} (i.e., without considering a specific rate-profile) polar cosets by ${\mathcal{P}_i(\mathcal I) = \mathcal{P}_i(\mathcal{I},\boldsymbol{T}=\boldsymbol{I}_N)}$ and $\mathcal{P}_i=\mathcal{P}_i(\mathcal I=\mathbb{Z}_N)$, respectively.

As shown in \cite{RowshanFormation}, any codeword in a universal coset $\mathcal{P}_i$ has at least the Hamming weight of the coset leader $\boldsymbol{g}_i$, i.e.,
\begin{equation}\label{eq:UTPdmin}
    \w\left(\boldsymbol{g}_i \oplus \bigoplus\limits_{h \in \mathcal{H}} \boldsymbol{g}_h\right) \geq \w(\boldsymbol{g}_i),
\end{equation}
for all $\mathcal H \subseteq \mathbb Z_{i+1:N}$.
Since $\Ci \subseteq \mathcal{P}_i$, from Eq.~(\ref{eq:UTPdmin}) follows that the codewords of $\Ci$ have at least the Hamming weight $\w(\boldsymbol{g}_i)$ of the coset leader $\boldsymbol{g}_i$.
In the case of plain polar codes, we find $\boldsymbol{g}_i \in \mathcal{P}_i(\mathcal{I}) \subseteq \mathcal{P}_i$. 
Thus, the minimum distance of a plain polar code $\mathcal{P}(\mathcal{I})$ is given as 
\begin{equation}\label{eq:dmin lower bound}
     \dmin(\mathcal{P}(\mathcal{I}))=\min_{i \in \mathcal{I}}\left\{\w(\boldsymbol{g}_i)\right\} = \wmin \leq \dmin(\C)
\end{equation}
and is a lower bound for the minimum distance of a \ac{PTPC} $\C$. %
In the following, the lowest potential codeword weight ${\wmin=\wmin(\mathcal{I})}$ is a constant dependent only on the (fixed) rate-profile $\mathcal{I}$ of the given \ac{PTPC} $\C$.
Furthermore, Eq.~(\ref{eq:UTPdmin}) implies that codewords $\boldsymbol{c} \in \C$ with weight $\wmin \leq \w(\boldsymbol{c}) < 2\wmin$ can only occur in cosets with indices 
\begin{equation*}
    i \in \mathcal{I}_{\wmin} = \left\{j \in \mathcal{I} \,\middle|\, \w(\boldsymbol{g}_j) = \wmin\right\},
\end{equation*}
since $\w(\boldsymbol{g}_h) = 2^{\w(h)}$ for $h \in \mathbb{Z}_N$ \cite{RowshanFormation}. 
Finally, from Eq.~(\ref{eq:UTPdmin}) follows that in the formation of a codeword $\boldsymbol{c} = \boldsymbol{u} \cdot \boldsymbol{G}_N$ with weight $\w(\boldsymbol{c}) < 2\wmin$, at least one row $\boldsymbol{g}_i$ with $i \in \mathcal{I}_{\wmin}$ is involved, as $\operatorname{supp}(\boldsymbol{u}) \cap \mathcal{I}_{\wmin} \neq \emptyset$.

Based on these insights, we now consider how to choose the pre-transform $\boldsymbol{T}$ to reduce the number of $\wmin$-weight codewords $\Awmin(\C)$ of a \ac{PTPC}, starting from the corresponding plain polar code $\mathcal{P}(\mathcal{I})$.
Our strategy is to increase the weight of the rows $\boldsymbol{g}_i$ with $i \in \mathcal{I}_{\wmin}$ by precoding them with subsequent frozen rows $\boldsymbol{g}_f$ where $f \in \mathcal{F} \setminus \mathbb{Z}_i$.
Thus, we partition the frozen rows into those that can increase the weight of the coset leader $\boldsymbol{g}_i$ as
\begin{equation}
        \mathcal{F}_i^* = \left\{f \in \mathcal{F} \setminus \mathbb{Z}_i\,\middle|\,\w(\boldsymbol{g}_i \oplus \boldsymbol{g}_f) > \w(\boldsymbol{g}_i)\right\}
\end{equation}
and those that cannot, $\mathcal{F}_i^\circ = \mathcal{F} \setminus \mathcal{F}_i^*$.
Following this, as introduced in \cite{ZunkerTreeIntersection} and \cite{Rowshan2023Impact}, we distinguish two types of cosets:
\begin{equation*}
    \Ci \mathrel{\widehat{=}} 
    \begin{cases}
        \emph{``pre-transformable"} & \text{ 
        if } \mathcal{F}_i^* \neq \emptyset\\
        \emph{``non pre-transformable"} & \text{ 
        if } \mathcal{F}_i^* = \emptyset.\\
    \end{cases}
\end{equation*}
In other words, pre-transformable cosets $\Ci$ contain one or more frozen rows $f>i$ with the property that the addition to the coset leader increases the Hamming weight.
Correspondingly, the set of indices of pre-transformable cosets is $\mathcal{I}^*_{\wmin} = \left\{i \in \mathcal I_{\wmin} \,\middle|\, \mathcal{F}_i^* \neq \emptyset\right\}$ and of non pre-transformable cosets is $\mathcal{I}^\circ_{\wmin} = \mathcal{I}_{\wmin} \setminus \mathcal{I}^*_{\wmin}$.
It was shown in \cite{ZunkerTreeIntersection} and \cite{Rowshan2023Impact} that the number of $\wmin$-weight codewords in the non pre-transformable cosets with indices $\mathcal{I}^\circ_{\wmin}$ cannot be reduced by pre-transformation.
This is intuitive since there are no frozen rows that can increase the weight of the respective coset leader.
Consequently, for the design of the pre-transformation, we only have to consider the pre-transformable cosets $\mathcal{I}^*_{\wmin}$.
Following this classification, we find the following theorem.
\begin{theorem}\label{thm:dmin=wmin}
    A pre-transformed polar code $\C$ with decreasing rate-profile $\mathcal{I}$ has the same minimum distance as the corresponding plain polar code $\mathcal P(\mathcal{I})$. %
    \label{thm:dmin}
\end{theorem}
\begin{proof}
    From the \textit{left swap} rule follows that every decreasing monomial code $\mathcal{P}(\mathcal {I})$ includes the most reliable row $\boldsymbol g_{i^\circ}$ of weight $\wmin=\w(\boldsymbol{g}_{i^\circ})$ with index 
    \begin{equation*}
        {i^\circ = \max\left(\mathcal I_{\wmin}\right) = N - \frac{N}{\wmin}}
    \end{equation*}    
    since $i^\circ \succcurlyeq i$ for all ${i \in \mathcal I_{\wmin}}$. By applying the \textit{binary domination} rule, find that $\mathbb{Z}_{i^\circ:N}=\left\{j \in \mathbb Z_N \,\middle|\, i^\circ \preccurlyeq j \right\} \subseteq \mathcal{I}$,
    and thus $\boldsymbol{g}_i \in \C$ for $i>i^\circ$.
    Subsequently, as there is no frozen row $\boldsymbol{g}_f$ with $f \in \mathcal{F}_i^* = \emptyset$ to add to $\boldsymbol g_{i^\circ}$, it follows that $\boldsymbol g_{i^\circ} \in \mathcal P(\mathcal{I},\boldsymbol{T})$ and $i^\circ \in \mathcal{I}^\circ_{\wmin} \neq \emptyset$.
    With the lower bound of Eq.~(\ref{eq:dmin lower bound}) we find that $\dmin(\C) = \dmin(\mathcal P(\mathcal I)) = \wmin$.
\end{proof}

\subsection{Selection of Row-Merges}
If we restrict ourselves to using each frozen bit at most once to precode low-weight generators, the pre-transformation is simply a repetition and we obtain a row-merged polar code $\mathcal{P}(\mathcal{I},\mathcal{R})$.
We propose a greedy algorithm to successively select the row-merge pairs $\mathcal{R}$.
In each iteration, we find the pair of $(i, f)$ for $i \in \mathcal{I}^*_{\wmin}$ and $f \in \mathcal{F}_i^*$, such that the $\dmin$ is increased the most and/or $A_{\dmin}$ is decreased the most.
That selecting rows $\mathcal{F}_i^*$ for precoding row $i$ directly reduces $\Admin$ can be explained using the tree-formalism introduced in \cite{ZunkerTreeIntersection}, where in this case, leaf nodes are formed in the message tree corresponding to the $\wmin$-weight codewords.
Only considering the rows $\mathcal{F}_i^*$ at the start of the design process showed faster convergence of the algorithm and smaller values of $\Admin$ are achieved.
If $\Admin$ cannot be reduced further by only considering the rows $\mathcal{F}_i^*$, also the rows $\mathcal{F}_i^\circ$ are taken into account, as they might able to decrease $\Admin$ further, even though they do not directly increase the weight of a $\wmin$-weight generator.

\begin{algorithm}[tbh]
    \small%
	\SetAlgoLined\LinesNumbered
	\SetKwInOut{Input}{Input}\SetKwInOut{Output}{Output}
	\Input{Rate-profile $\mathcal I$}
	\Output{Row-merges $\mathcal{R}$}
    $\mathcal D \gets \emptyset,\;\mathcal R \gets \emptyset$\;
    $d_{\smash{\mathrm{min}}}^* \gets \dmin(\mathcal{P}(\mathcal{I},\mathcal{R})),\;A_{\smash{d_\mathrm{min}}}^* \gets \Admin(\mathcal{P}(\mathcal{I},\mathcal{R}))$\;
    $\operatorname{increase\_weight} \gets \mathrm{\textbf{True}}$\;
    \While{\rm{\textbf{True}}}{
        $d_{\smash{\mathrm{min}}}^{**} \gets d_{\smash{\mathrm{min}}}^*$, $A^{**}_{\smash{d_\mathrm{min}}} \gets A_{\smash{d_\mathrm{min}}}^*$\;\vspace{0.025cm}
        \ForEach{$i \in \mathcal{I}_{\wmin}^*$\vspace{0.025cm}}{ 
            \uIf{$\operatorname{increase\_weight}$}{
                $\mathcal{D}' \gets \mathcal{F}_i^* \setminus \mathcal{D}$\;
            }
            \Else{
                $\mathcal{D}' \gets (\mathcal{F}_i^* \cup \mathcal{F}_i^\circ) \setminus \mathcal{D}$\;
            }
            \ForEach{$f \in \mathcal{D}'$ }{ \label{alg:row_merge:forf}
                $\mathcal{R}' \gets \mathcal{R} \cup \{(i,f)\})$\;
                $\dmin \gets \dmin(\mathcal{P}(\mathcal{I},\mathcal{R}'),\;
                \Admin \gets \Admin(\mathcal{P}(\mathcal{I},\mathcal{R}')$\;   
                \If{$(\dmin = d_{\smash{\mathrm{min}}}^*$ \rm{\textbf{and}} $\Admin < A_{\smash{d_\mathrm{min}}}^*)$ \rm{\textbf{or}} \break
                $\dmin > d_{\smash{\mathrm{min}}}^*$}{
                    $d_{\smash{\mathrm{min}}}^{*} \gets \dmin,\,A^{*}_{\smash{d_\mathrm{min}}} \gets \Admin,\,i^* \gets i,\,f^* \gets f$\; 
                }
            }
        }
        \uIf{$A_{\smash{d_\mathrm{min}}}^* < A^{**}_{\smash{d_\mathrm{min}}}$\vspace{0.025cm}}{
            $\mathcal{D} \gets \mathcal{D} \cup \{f^*\},\;\mathcal{R} \gets \mathcal{R} \cup \{(i^*,f^*)\}$\;
        }
        \uElseIf{$\operatorname{increase\_weight}$}{
                $\operatorname{increase\_weight} \gets \mathrm{\textbf{False}}$\;
            }
        \lElse{\Return $\mathcal{R}$}
    }
	\caption{\footnotesize $\operatorname{merge\_rows}(\mathcal{I})$}
    \label{alg:row_merge}
\end{algorithm}

Alg.~\ref{alg:row_merge} lists the pseudo-code of the proposed procedure.
Which of the two phases of the search is currently active is indicated by the variable ``$\operatorname{increase\_weight}$''.
The functions ``$\dmin(\mathcal{P}(\mathcal{I},\mathcal{R}))$'' and ``$\Admin(\mathcal{P}(\mathcal{I},\mathcal{R}))$'' retrieve the minimum distance $\dmin$ and the number of minimum-weight codewords $A_{\dmin}$ for the row-merged polar code with rate-profile $\mathcal{I}$ and row-merges $\mathcal{R}$, respectively.
For these functions, we use the implementation given in \cite[Alg.~1]{ZunkerTreeIntersection}. 
If all minimum-weight codewords of the plain polar code have been eliminated, we revert to the more complex method proposed in \cite{PartialEnumPAC}, generalized as in \cite{ZunkerTreeIntersection} to the same recursive exploration of polar cosets to find which sub-trees have to be explored to count the codewords of a certain weight.
Note that the general algorithm presented in \cite{Miloslavskaya2022partialweightdistribution} can also be used, but has a much higher computational complexity in the case of a polar-like rate-profile.

\begin{remark}\label{rk:simplification}\normalfont
To reduce the computational complexity of Alg.~\ref{alg:row_merge}, one can limit the exploration space for each coset in Line~\ref{alg:row_merge:forf} to only a subset of $\mathcal{D}'$.
Namely, we found that restriction to the subsequent two frozen indices (i.e., the two smallest indices in $\mathcal{D}'$) results in similarly performing and often the same row-merged polar codes. 
This is intuitive, considering the tree-formalism introduced in \cite{ZunkerTreeIntersection}. 
The message sub-trees at levels with smaller indices are larger. 
Thus, a pre-transformation using a frozen bit with a smaller index can achieve a larger reduction of $\Awmin$ with fewer sub-trees to eliminate.
This is backed up by concrete results given in Sec.~\ref{sec:results}. 
\end{remark}

\subsection{Design of the Rate Profile}
\label{ssec:rateproofiledesign}
For \ac{PAC} codes under sequential (Fano) decoding, as proposed in \cite{arikan2019pac}, \ac{ML} performance is the main design criterion and hence, the \ac{RM} rate-profile is a reasonable choice.
However, for fixed-latency \ac{SCL} decoding, \ac{SC} decodability, i.e., the bit error probabilities of the information bit channels using the \ac{SC} schedule, plays an important role.
\Ac{DE} is a well-known method to construct good polar codes for \ac{SC}-based decoding \cite{constructDE} and, thus, is the basis of our construction.
\Ac{DE} tracks the probability densities of the channel output at a given \ac{SNR} (the \emph{design SNR} ${E_\mathrm{s}}/{N_0}$) through the decoding graph and retrieves the error-probabilities of all synthetic channels. The $K$ most reliable channels are then selected as the information set $\mathcal{I}$. While this optimizes the performance of \ac{SC} decoding at the design \ac{SNR}, other code designs can be also generated by selecting different design \acp{SNR}. 
For example, it is well-known that a high design \ac{SNR} results in codes with large $\dmin$, ultimately producing \ac{RM} codes in the limit.

\subsubsection{Decreasing Rate-Profiles}
From Thm.~\ref{thm:dmin} we know that for decreasing rate profiles $\mathcal{I}$, pre-transformations (including row-merges) cannot increase the minimum distance. 
Hence, in the short block length regime, where performance is mostly limited by the weight spectrum, the plain polar code must already have a sufficiently large $\dmin$.
Therefore, the design \ac{SNR} of \ac{DE} should be selected as low as possible (for good decodability), while still reaching the desired $\dmin$. In particular, we find good codes with code dimensions slightly below \ac{RM} codes. The \ac{DE} design can then reach the $\dmin$ of the next larger \ac{RM} code, but exclude the (typically early) synthetic bit channels that have a high bit error probability.

\begin{example}\label{ex:P60128}
\normalfont Consider the case ${N=128}$, ${K=60}$. The next larger \ac{RM} code with the same block length is $\mathcal{RM}(3,7)$ with ${K=64}$ and ${\dmin=16}$. 
Hence, it is feasible to construct a polar code with ${K=60}$ and ${\dmin=16}$.
We find that a design \ac{SNR} of at least \SI{2.9}{\dB} is required to achieve ${\dmin=16}$. The resulting code has the minimal information set ${\mathcal I_\mathrm{min}=\{29,43,71\}}$.
\end{example}

\begin{example}\label{ex:P75256}
\normalfont Consider a medium block length of $N=256$ and code dimension $K=75$.
The next larger \ac{RM} code is $\mathcal{RM}(3,8)$ with code dimension $K=93$ and $d_\mathrm{min}=32$.
Thus, there is plenty of leeway to design a rate-profile better suited for \ac{SC}-based decoding. Namely, \ac{DE} at ${E_\mathrm{s}}/{N_0} = \qty{0.1}{dB}$ results in a code with $d_\mathrm{min}=32$ and $\mathcal I_\mathrm{min}=\{63,115,157,167\}$. In comparison, the 5G polar code (without \ac{CRC}) has only $\dmin=16$.
\end{example}

\subsubsection{Non-Decreasing Rate-Profiles}
While the above method is effective in designing good decreasing rate-profiles, it is not feasible in the following cases:
\begin{itemize}
    \item For long block lengths or code dimensions too close to an \ac{RM} code, where \ac{SC} decodability is the main limiting factor (and design \ac{SNR} must be lowered, decreasing $\dmin$)
    \item When the desired $\dmin$ cannot be achieved by \ac{DE}
\end{itemize}

Hence, the pre-transform must be used to increase the minimum distance of the plain polar code.
We know from Thm.~\ref{thm:dmin} that the rate-profile cannot be decreasing in this case.
In particular, the result from Thm.~\ref{thm:dmin} is due to the last minimum-weight generators in a decreasing rate profile, which cannot be pre-transformed (as there are no more frozen bits below). 
To circumvent this, we propose to freeze the last $\kappa$ non pre-transformable, minimum-weight generators, which can then be used in row-merges to precode rows above.
The resulting rate loss has to be counteracted by unfreezing other rows, i.e., designing a code using \ac{DE} with $K'=K+\kappa$. In practice, this method can be implemented iteratively, as listed in Alg.~\ref{alg:rate_profile_design}. Here, $\boldsymbol{q}$ denotes the synthetic channel ranking from the best ($q_0$) to the worst ($q_{N-1}$) bit channel.

\begin{algorithm}[tbh]
    \small%
	\SetAlgoLined\LinesNumbered
	\SetKwInOut{Input}{Input}\SetKwInOut{Output}{Output}
	\Input{Channel ranking $\boldsymbol{q}$, code length $N$,\newline code dimension $K$, target minimum distance $d_{\smash{\mathrm{min}}}^*$}
	\Output{Rate-profile $\mathcal{I}$, set of row-merges $\mathcal{R}$}
    $\mathcal{I} \gets  \{q_i\,|\, i \in \mathbb{Z}_K\},\;\mathcal{R} \gets \emptyset$\;%
    $\kappa \gets 0$\;
    \While{$\dmin(\mathcal{P}(\mathcal{I},\mathcal{R})) < d_{\smash{\mathrm{min}}}^*$}{
        \lIf{$K+\kappa=N$}{\Return \emph{design failure}}
        $\mathcal{I} \gets (\mathcal{I} \setminus \{\max(\mathcal{I}_{\wmin})\}) \cup \{q_{K+\kappa}\}$\;
        $\mathcal{R} \gets \operatorname{merge\_rows}(\mathcal{I})$\;
        $\kappa \gets \kappa+1$\;
    }
    \Return $\mathcal I,\;\mathcal{R}$\;
	\caption{\footnotesize Non-decreasing row-merged polar code design}
    \label{alg:rate_profile_design}
\end{algorithm}

\begin{example}\label{ex:P5121024}
\normalfont Let $N=1024$ and $K=512$. The next larger \ac{RM} code is $\mathcal{RM}(5,10)$ with code dimension ${K=638}$ and ${\dmin=32}$.
A minimum design SNR of \qty{2.2}{dB} is required to obtain a ${\dmin=32}$ polar code by \ac{DE}.
This design limits the decodability in the low to medium \ac{SNR} range.
Therefore, if the operating point is not supposed to be at a very low \ac{BLER}, the design \ac{SNR} should be reduced.
\ac{DE} at design \ac{SNR} of $\qty{0}{dB}$ results in ${\left|\mathcal{I}^\circ_{\wmin}\right|=12}$ non-pre-transformable rows of weight $\wmin=16$, namely the last 12 weight-16 rows. We find that freezing the last $\kappa=7$ of these rows allows the row-merge precoder to eliminate all weight-16 codewords as the remaining weight-16 rows become pre-transformable. %
To counteract the rate-loss, the code dimension of the base code design is increased to $K'=512+7$. 
\end{example}

Note that such a design is similar to conventional \ac{CRC}-aided polar codes, where high-index information bits are turned into dynamic frozen bits. However, the explicit selection of these vulnerable positions results generally in better code designs than a \ac{CRC} in the very last information bits, as it is easier to increase $\dmin$ and bit channels with higher error probability are eliminated.
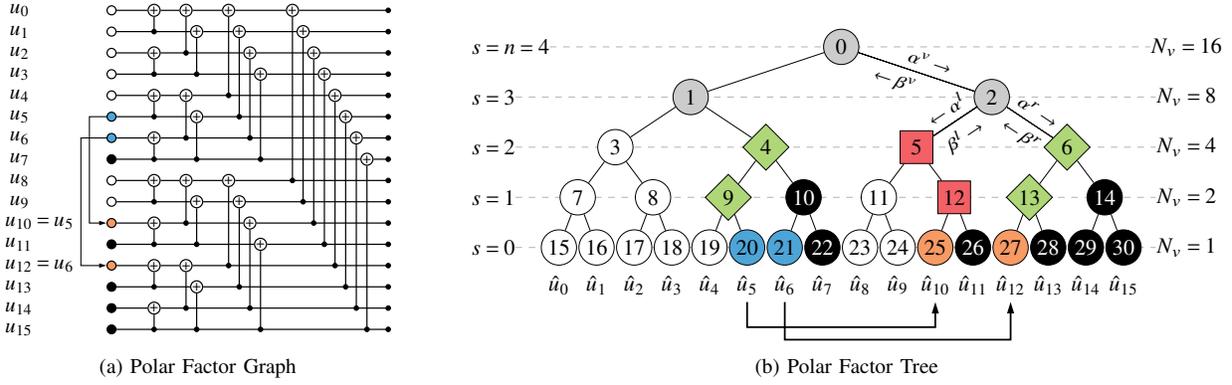
\begin{figure*}
	\centering
    \subfloat[\footnotesize Polar Factor Graph]{%
        \resizebox{0.29\columnwidth}{!}%
        {\colorlet{R0}{white}%
\colorlet{R1}{black}%
\colorlet{mI}{mittelblau!70!white}%
\colorlet{dF}{orange!70!white}%
\newcommand\circleSize{6pt}%
\begin{tikzpicture}[thick]%
    \tikzstyle{frozennode} = [circle, minimum size=\circleSize, inner sep=0pt,
        fill=R0, draw=black]%
    \tikzstyle{normalnode} = [circle, minimum size=\circleSize, inner sep=0pt,
        fill=R1, draw=black, line width=0]%
    \tikzstyle{dynfrozen}  = [circle, minimum size=\circleSize, inner sep=0pt,
        fill=dF, draw=black] %
    \tikzstyle{mergedinfo} = [circle, minimum size=\circleSize, inner sep=0pt,
        fill=mI, draw=black] %
    \tikzstyle{xor} = [circle,draw,inner sep=-1pt, minimum size=8pt] %
    \tikzstyle{dot} = [dspnodefull,minimum size=3pt,inner sep=0pt]%
    \node[frozennode] at (0, 0.0)  (u0) {};
    \node[frozennode] at (0,-0.5)  (u1) {};
    \node[frozennode] at (0,-1.0)  (u2) {};
    \node[frozennode] at (0,-1.5)  (u3) {};
    \node[frozennode] at (0,-2.0)  (u4) {};
    \node[mergedinfo] at (0,-2.5)  (u5) {};
    \node[mergedinfo] at (0,-3.0)  (u6) {};
    \node[normalnode] at (0,-3.5)  (u7) {};
    \node[frozennode] at (0,-4.0)  (u8) {};
    \node[frozennode] at (0,-4.5)  (u9) {};
    \node[dynfrozen]  at (0,-5.0) (u10) {};
    \node[normalnode] at (0,-5.5) (u11) {};
    \node[dynfrozen]  at (0,-6.0) (u12) {};
    \node[normalnode] at (0,-6.5) (u13) {};
    \node[normalnode] at (0,-7.0) (u14) {};
    \node[normalnode] at (0,-7.5) (u15) {};
    \node[xor]  (n00) at (1.0, 0.0) {+} edge [-]  (u0);
    \node[dot]  (n10) at (1.0,-0.5) {}  edge [-]  (u1);
    \node[xor]  (n20) at (1.0,-1.0) {+} edge [-]  (u2);
    \node[dot]  (n30) at (1.0,-1.5) {}  edge [-]  (u3);
    \node[xor]  (n40) at (1.0,-2.0) {+} edge [-]  (u4);
    \node[dot]  (n50) at (1.0,-2.5) {}  edge [-]  (u5);
    \node[xor]  (n60) at (1.0,-3.0) {+} edge [-]  (u6);
    \node[dot]  (n70) at (1.0,-3.5) {}  edge [-]  (u7);
    \node[xor]  (n80) at (1.0,-4.0) {+} edge [-]  (u8);
    \node[dot]  (n90) at (1.0,-4.5) {}  edge [-]  (u9);
    \node[xor] (n100) at (1.0,-5.0) {+} edge [-] (u10);
    \node[dot] (n110) at (1.0,-5.5) {}  edge [-] (u11);
    \node[xor] (n120) at (1.0,-6.0) {+} edge [-] (u12);
    \node[dot] (n130) at (1.0,-6.5) {}  edge [-] (u13);
    \node[xor] (n140) at (1.0,-7.0) {+} edge [-] (u14);
    \node[dot] (n150) at (1.0,-7.5) {}  edge [-] (u15);
    
    \draw[-]  (n00)-- (n10);
    \draw[-]  (n20)-- (n30);
    \draw[-]  (n40)-- (n50);
    \draw[-]  (n60)-- (n70);
    \draw[-]  (n80)-- (n90);
    \draw[-] (n100)--(n110);
    \draw[-] (n120)--(n130);
    \draw[-] (n140)--(n150);
    \node[xor]  (n01) at (1.75, 0.0) {+} edge [-]  (n00);
    \node[xor]  (n11) at (2.00,-0.5) {+} edge [-]  (n10);
    \node[dot]  (n21) at (1.75,-1.0) {}  edge [-]  (n20);
    \node[dot]  (n31) at (2.00,-1.5) {}  edge [-]  (n30);%
    \node[xor]  (n41) at (1.75,-2.0) {+} edge [-]  (n40);
    \node[xor]  (n51) at (2.00,-2.5) {+} edge [-]  (n50);
    \node[dot]  (n61) at (1.75,-3.0) {}  edge [-]  (n60);
    \node[dot]  (n71) at (2.00,-3.5) {}  edge [-]  (n70);
    \node[xor]  (n81) at (1.75,-4.0) {+} edge [-]  (n80);
    \node[xor]  (n91) at (2.00,-4.5) {+} edge [-]  (n90);
    \node[dot] (n101) at (1.75,-5.0) {}  edge [-] (n100);
    \node[dot] (n111) at (2.00,-5.5) {}  edge [-] (n110);%
    \node[xor] (n121) at (1.75,-6.0) {+} edge [-] (n120);
    \node[xor] (n131) at (2.00,-6.5) {+} edge [-] (n130);
    \node[dot] (n141) at (1.75,-7.0) {}  edge [-] (n140);
    \node[dot] (n151) at (2.00,-7.5) {}  edge [-] (n150);
    
    \draw[-]  (n01)-- (n21);
    \draw[-]  (n11)-- (n31);
    \draw[-]  (n41)-- (n61);
    \draw[-]  (n51)-- (n71);
    \draw[-]  (n81)--(n101);
    \draw[-]  (n91)--(n111);
    \draw[-] (n121)--(n141);
    \draw[-] (n131)--(n151);
    \node[xor]  (n02) at (2.75, 0.0) {+} edge [-]  (n01);
    \node[xor]  (n12) at (3.00,-0.5) {+} edge [-]  (n11);
    \node[xor]  (n22) at (3.25,-1.0) {+} edge [-]  (n21);
    \node[xor]  (n32) at (3.50,-1.5) {+} edge [-]  (n31);
    \node[dot]  (n42) at (2.75,-2.0) {}  edge [-]  (n41);
    \node[dot]  (n52) at (3.00,-2.5) {}  edge [-]  (n51);
    \node[dot]  (n62) at (3.25,-3.0) {}  edge [-]  (n61);
    \node[dot]  (n72) at (3.50,-3.5) {}  edge [-]  (n71);%
    \node[xor]  (n82) at (2.75,-4.0) {+} edge [-]  (n81);
    \node[xor]  (n92) at (3.00,-4.5) {+} edge [-]  (n91);
    \node[xor] (n102) at (3.25,-5.0) {+} edge [-] (n101);
    \node[xor] (n112) at (3.50,-5.5) {+} edge [-] (n111);
    \node[dot] (n122) at (2.75,-6.0) {}  edge [-] (n121);
    \node[dot] (n132) at (3.00,-6.5) {}  edge [-] (n131);
    \node[dot] (n142) at (3.25,-7.0) {}  edge [-] (n141);
    \node[dot] (n152) at (3.50,-7.5) {}  edge [-] (n151);
    
    \draw[-]  (n02)-- (n42);
    \draw[-]  (n12)-- (n52);
    \draw[-]  (n22)-- (n62);
    \draw[-]  (n32)-- (n72);%
    \draw[-]  (n82)--(n122);
    \draw[-]  (n92)--(n132);
    \draw[-] (n102)--(n142);
    \draw[-] (n112)--(n152);
    \node[xor]  (n03) at (4.25, 0.0) {+} edge [-]  (n02);
    \node[xor]  (n13) at (4.50,-0.5) {+} edge [-]  (n12);
    \node[xor]  (n23) at (4.75,-1.0) {+} edge [-]  (n22);
    \node[xor]  (n33) at (5.00,-1.5) {+} edge [-]  (n32);
    \node[xor]  (n43) at (5.25,-2.0) {+} edge [-]  (n42);
    \node[xor]  (n53) at (5.50,-2.5) {+} edge [-]  (n52);
    \node[xor]  (n63) at (5.75,-3.0) {+} edge [-]  (n62);
    \node[xor]  (n73) at (6.00,-3.5) {+} edge [-]  (n72);
    \node[dot]  (n83) at (4.25,-4.0) { } edge [-]  (n82);
    \node[dot]  (n93) at (4.50,-4.5) { } edge [-]  (n92);
    \node[dot] (n103) at (4.75,-5.0) { } edge [-] (n102);
    \node[dot] (n113) at (5.00,-5.5) { } edge [-] (n112);
    \node[dot] (n123) at (5.25,-6.0) { } edge [-] (n122);
    \node[dot] (n133) at (5.50,-6.5) { } edge [-] (n132);
    \node[dot] (n143) at (5.75,-7.0) { } edge [-] (n142);
    \node[dot] (n153) at (6.00,-7.5) { } edge [-] (n152);

    \draw[-]  (n03)-- (n83);
    \draw[-]  (n13)-- (n93);
    \draw[-]  (n23)--(n103);
    \draw[-]  (n33)--(n113);
    \draw[-]  (n43)--(n123);
    \draw[-]  (n53)--(n133);
    \draw[-]  (n63)--(n143);
    \draw[-]  (n73)--(n153);
    \node[dot]  (end0) at (6.50, 0.0) {} edge [-]  (n03);
    \node[dot]  (end1) at (6.50,-0.5) {} edge [-]  (n13);
    \node[dot]  (end2) at (6.50,-1.0) {} edge [-]  (n23);
    \node[dot]  (end3) at (6.50,-1.5) {} edge [-]  (n33);
    \node[dot]  (end4) at (6.50,-2.0) {} edge [-]  (n43);
    \node[dot]  (end5) at (6.50,-2.5) {} edge [-]  (n53);
    \node[dot]  (end6) at (6.50,-3.0) {} edge [-]  (n63);
    \node[dot]  (end7) at (6.50,-3.5) {} edge [-]  (n73);
    \node[dot]  (end8) at (6.50,-4.0) {} edge [-]  (n83);
    \node[dot]  (end9) at (6.50,-4.5) {} edge [-]  (n93);
    \node[dot] (end10) at (6.50,-5.0) {} edge [-] (n103);
    \node[dot] (end11) at (6.50,-5.5) {} edge [-] (n113);
    \node[dot] (end12) at (6.50,-6.0) {} edge [-] (n123);
    \node[dot] (end13) at (6.50,-6.5) {} edge [-] (n133);
    \node[dot] (end14) at (6.50,-7.0) {} edge [-] (n143);
    \node[dot] (end15) at (6.50,-7.5) {} edge [-] (n153);

\newcommand\xPos{-2.5}
    \node (in0)  at (\xPos, 0.0)[label={[align=left,shift=(in0.west),anchor=west, font=\Large]$\;{u}_{0}$}] {};
    \node (in1)  at (\xPos,-0.5)[label={[align=left,shift=(in1.west),anchor=west, font=\Large]$\;{u}_{1}$}] {};
    \node (in2)  at (\xPos,-1.0)[label={[align=left,shift=(in2.west),anchor=west, font=\Large]$\;{u}_{2}$}] {};
    \node (in3)  at (\xPos,-1.5)[label={[align=left,shift=(in3.west),anchor=west, font=\Large]$\;{u}_{3}$}] {};
    \node (in4)  at (\xPos,-2.0)[label={[align=left,shift=(in4.west),anchor=west, font=\Large]$\;{u}_{4}$}] {};
    \node (in5)  at (\xPos,-2.5)[label={[align=left,shift=(in5.west),anchor=west, font=\Large]$\;{u}_{5}$}] {};
    \node (in6)  at (\xPos,-3.0)[label={[align=left,shift=(in6.west),anchor=west, font=\Large]$\;{u}_{6}$}] {};
    \node (in7)  at (\xPos,-3.5)[label={[align=left,shift=(in7.west),anchor=west, font=\Large]$\;{u}_{7}$}] {};
    \node (in8)  at (\xPos,-4.0)[label={[align=left,shift=(in8.west),anchor=west, font=\Large]$\;{u}_{8}$}] {};
    \node (in8)  at (\xPos,-4.5)[label={[align=left,shift=(in8.west),anchor=west, font=\Large]$\;{u}_{9}$}] {};
    \node (in10) at (\xPos,-5.0)[label={[align=left,shift=(in10.west),anchor=west, font=\Large]${u}_{10}={u}_{5}$}] {};
    \node (in11) at (\xPos,-5.5)[label={[align=left,shift=(in11.west),anchor=west, font=\Large]${u}_{11}$}] {};
    \node (in12) at (\xPos,-6.0)[label={[align=left,shift=(in12.west),anchor=west, font=\Large]${u}_{12}={u}_{6}$}] {};
    \node (in13) at (\xPos,-6.5)[label={[align=left,shift=(in13.west),anchor=west, font=\Large]${u}_{13}$}] {};
    \node (in14) at (\xPos,-7.0)[label={[align=left,shift=(in14.west),anchor=west, font=\Large]${u}_{14}$}] {};
    \node (in15) at (\xPos,-7.5)[label={[align=left,shift=(in15.west),anchor=west, font=\Large]${u}_{15}$}] {};

    \draw[dspconn, black] (u5.west) -- +(-0.4cm, 0) |- (u10); %
    \draw[dspconn, black] (u6.west) -- +(-0.6cm, 0) |- (u12);

\end{tikzpicture}}\label{subfig:polarGraph}
    }
    \hfil
    \subfloat[\footnotesize Polar Factor Tree]{%
        \resizebox{0.55\columnwidth}{!}%
        {\colorlet{REP}{rot!70!}%
\colorlet{SPC}{apfelgruen!70!}%
\colorlet{R0}{white}%
\colorlet{R1}{black}%
\colorlet{mI}{mittelblau!70!}%
\colorlet{dF}{orange!70!}%
\begin{tikzpicture}[%
    -,sloped,
    level 1/.style={sibling distance=4.8cm, level distance=0.8cm},
    level 2/.style={sibling distance=2.4cm},
    level 3/.style={sibling distance=1.2cm},
    level 4/.style={sibling distance=0.6cm},
    level 5/.style={sibling distance=0.3cm},
    level 6/.style={sibling distance=0.3cm},
    level 5/.style={sibling distance=0.3cm, level distance=0.6cm},
    square/.style={regular polygon,regular polygon sides=4},
    triangleV/.style={regular polygon, regular polygon sides=3,
            shape border rotate=180 },
    pentagonA/.style={regular polygon, regular polygon sides=5},
    pentagonV/.style={regular polygon, regular polygon sides=5,
            shape border rotate=180 },
    treenode/.style={align=center, circle, draw=black, fill=mittelgrau!40,
            inner sep = 0pt, minimum size=16pt},
    node_info/.style={treenode, circle, white, draw=black, fill=R1},
    node_frozen/.style={treenode, circle, black, fill=R0},
    node_r_info/.style={treenode, circle, black, draw=black, fill=mI},
    node_d_frozen/.style={treenode, circle, black, fill=dF},
    node_rep/.style={treenode, square, black, fill=REP, minimum size=21pt}, 
    node_spc/.style={treenode, diamond, black, fill=SPC, minimum size=21pt},
    leaf/.style={align=center},
    ]%
    \node (A) at (-5.9,0)[inner sep=0,black,anchor=west] {$s=n=4$};
    \node (B) at ( 5.5,0)[inner sep=0] {$N_v=16$};
    \draw[dashed,black!30](A) -- (B);

    \node (A) at (-5.9,-0.8)[inner sep=0,black,anchor=west] {$s=3$};
    \node (B) at ( 5.5,-0.8)[inner sep=0] {$N_v=8$};
    \draw[dashed,black!30](A) -- (B);

    \node (A) at (-5.9,-1.6)[inner sep=0,black,anchor=west] {$s=2$};
    \node (B) at ( 5.5,-1.6)[inner sep=0] {$N_v=4$};
    \draw[dashed,black!30](A) -- (B);

    \node (A) at (-5.9,-2.4)[inner sep=0,black,anchor=west] {$s=1$};
    \node (B) at ( 5.5,-2.4)[inner sep=0] {$N_v=2$};
    \draw[dashed,black!30](A) -- (B);

    \node (A) at (-5.9,-3.2)[inner sep=0,black,anchor=west] {$s=0$};
    \node (B) at ( 5.5,-3.2)[inner sep=0] {$N_v=1$};
    \draw[dashed,black!30](A) -- (B);

    \begin{scope}[execute at begin node=$, execute at end node=$]
            \node at (0,0) [treenode]{0}
            child{
                node[treenode]{1}
                    child{node[node_frozen]{3}
                            child{node[node_frozen]{7}
                                    child{node[node_frozen]{15}
                                            child{node[leaf]{$\hat{u}_0$}edge from parent[draw=none]}
                                    }
                                    child{node[node_frozen]{16}
                                            child{node[leaf]{$\hat{u}_1$}edge from parent[draw=none]}
                                    }
                            }
                            child{node[node_frozen]{8}
                                    child{node[node_frozen]{17}
                                            child{node[leaf]{$\hat{u}_2$}edge from parent[draw=none]}
                                    }
                                    child{node[node_frozen]{18}
                                            child{node[leaf]{$\hat{u}_3$}edge from parent[draw=none]}
                                    }
                            }
                    }
                    child{node[node_spc]{4}
                            child{node[node_spc]{9}
                                    child{node[node_frozen]{19}
                                            child{node[leaf]{$\hat{u}_4$}edge from parent[draw=none]}
                                    }
                                    child{node[node_r_info]{20}
                                            child{node[leaf](5){$\hat{u}_{5}$}edge from parent[draw=none]}
                                    }
                            }
                            child{node[node_info]{10}
                                    child{node[node_r_info]{21}
                                            child{node[leaf](6){$\hat{u}_{6}$}edge from parent[draw=none]}
                                    }
                                    child{node[node_info]{22}
                                            child{node[leaf]{$\hat{u}_{7}$}edge from parent[draw=none]}
                                    }
                            }
                    }
            }
        child{node[treenode]{2}
                    child{node[node_rep]{5}
                            child{node[node_frozen]{11}
                                    child{node[node_frozen]{23}
                                            child{node[leaf]{$\hat{u}_{8}$}edge from parent[draw=none]}
                                    }
                                    child{node[node_frozen]{24}
                                            child{node[leaf]{$\hat{u}_{9}$}edge from parent[draw=none]}
                                    }
                            }
                            child{node[node_rep]{12}
                                    child{node[node_d_frozen]{25}
                                            child{node[leaf](10){$\hat{u}_{10}$}edge from parent[draw=none]}
                                    }
                                    child{node[node_info]{26}
                                            child{node[leaf]{$\hat{u}_{11}$}edge from parent[draw=none]}
                                    }
                            }
                            {
                                    edge from parent node[above, font=\footnotesize]{\leftarrow \alpha^l}
                                    edge from parent node[below, font=\footnotesize]{\mathbf{\beta}^l \rightarrow}
                            }
                    }
                    child{node[node_spc]{6}
                            child{node[node_spc]{13}
                                    child{node[node_d_frozen]{27}
                                            child{node[leaf](12){$\hat{u}_{12}$}edge from parent[draw=none]}
                                    }
                                    child{node[node_info]{28}
                                            child{node[leaf]{$\hat{u}_{13}$}edge from parent[draw=none]}
                                    }
                            }
                            child{node[node_info]{14}
                                    child{node[node_info]{29}
                                            child{node[leaf]{$\hat{u}_{14}$}edge from parent[draw=none]}
                                    }
                                    child{node[node_info]{30}
                                            child{node[leaf]{$\hat{u}_{15}$}edge from parent[draw=none]}
                                    }
                            }
                            edge from parent node[above, font=\footnotesize]{\alpha^r \rightarrow}
                            edge from parent node[below, font=\footnotesize]{\leftarrow \beta^r}
                    }
                    edge from parent node[above, font=\footnotesize]{\quad \alpha^v \rightarrow}
                    edge from parent node[below, font=\footnotesize]{\leftarrow \beta^v \qquad}
            };

            \draw[dspconn, black, thick] (5.south) -- +(0, -0.4cm) -| (10); %
            \draw[dspconn, black, thick] (6.south) -- +(0, -0.6cm) -| (12); %

    \end{scope}

\end{tikzpicture}%
}\label{subfig:polarTree}
    }
	\caption{\footnotesize
        (a) Polar factor graph and (b) the corresponding \acs{PFT} for a
        $\mathcal C(16,7)$ with $\mathcal{R}=\{(5,10), (6,12)\}$. Simple 
        information and static frozen bits are colored in black and white,
        repeated information and dynamic frozen bits are marked in blue and 
        orange, respectively. Rate-0 and Rate-1 nodes are filled in white and
        black, \acs{REP} and \acs{SPC} nodes are shown as red squares and green
        rhombs, respectively.
    }
	\label{fig:polarGraphTree}
\end{figure*}
\section{Fast Simplified Successive Cancellation %
List Decoding with Dynamic Frozen Bits}
\label{sec:fastSSCLdf}

As \acp{PTPC} are decoded by \acf{SCL} decoding, we first review
the \ac{SCL} decoding algorithm with its state-of-the-art optimizations and the 
needed adaptions with respect to the pre-transformation. In the second part of
this section, the additional calculations needed for \ac{FSSCL} decoding with 
dynamic frozen bits are described.

\subsection{State of the Art}
\subsubsection{SCL Decoding}\hfill\break
\ac{SC}-based decoding is typically represented as traversal of a balanced 
binary tree, the \ac{PFT}~\cite{alamdar2011sscTree}, which is derived from the
polar factor graph. An example for a notional $\mathcal C(16,7)$ is given in 
\autoref{fig:polarGraphTree}. The root node receives the channel \acp{LLR},
defined as ${\operatorname{LLR}(y_i) = \log(\operatorname{P}(y_i | x_i=0) /
\operatorname{P}(y_i | x_i=1))}$, calculated for all $N$~received channel 
values~$\boldsymbol{y}$ with $i\in\mathbb Z_N$. A node~$v$ in layer $s$
receives a vector~$\boldsymbol{\alpha}^v$ of $N_v=2^s$~\acp{LLR} from its
parent node in layer $s+1$. The messages passed to the left and the right child
nodes, $\boldsymbol{\alpha}^l$ and $\boldsymbol{\alpha}^r$, are calculated by
$f$- and $g$-functions,
respectively. The partial-sum vector~$\boldsymbol{\beta}^v$ is calculated by
the $h$-function combining $\boldsymbol{\beta}^l$ and $\boldsymbol{\beta}^r$ 
and returned to the parent node. Thus, in \ac{SC}-based decoding, the \ac{PFT}
is traversed depth first with an inherent priority to the left child. In the
$N$ leaf nodes, $\boldsymbol{\beta}^v$ is set to the value of the estimated 
bit~$\hat{u}_i$ with $i\in \mathbb Z_N$ as
\begin{equation}
    \hat{u}_i =
    \begin{cases}
        0, &\text{if } i \in \mathcal{F} \text{ or } \alpha^v_i\geq 0 \\
        1, &\text{otherwise.}
    \end{cases}
    \label{eq:u_i}
\end{equation}

In \ac{SCL} decoding \cite{talvardyList}, lists of $L$~vectors are passed among 
the nodes of the \ac{PFT} instead of only passing vectors~$\boldsymbol{\alpha}$ 
and $\boldsymbol{\beta}$. Therefore, ${l\in\mathbb Z_L}$ is introduced
to index the $L$ concurrent paths, \eg the \ac{LLR} vectors of a node~$v$ are
denoted as~$\boldsymbol{\alpha}^v_l$. For an information bit estimation in 
layer~$s=0$ of the \ac{PFT}, the decoding path splits, \ie both possible 
values, $0$ and $1$, are considered. Consequently, an information bit 
estimation doubles the number of decoding paths. Unreliable paths must be 
rejected by sorting, when the number of paths exceeds $L$. For this purpose, 
every bit estimation updates a \ac{PM} to rate the reliability of each path. 
These \acp{PM} are initialized with $0$ in \ac{LLR}-based \ac{SCL} decoding 
\cite{balatsoukas2015LLRSCL}. The \ac{PM} of a path with
index~$p \in \mathbb Z_{2L}$ proceeding input path~${l\in\mathbb Z_L}$ is
updated for the $i$-th bit estimation by
\begin{equation}
    \text{PM}_p^{i} =
    \begin{cases}
        \text{PM}_l^{i-1}
        + \left|\alpha^v_{l,i} \right|,
            &\text{if } \beta^v_{l,i} \neq \text{HD}(\alpha^v_{l,i})\\
        \text{PM}_l^{i-1},
            &\text{otherwise,}
    \end{cases}
    \label{eq:pm}
\end{equation}
with \ac{HD} on~$\alpha^v_{l,i}$ defined as
\begin{equation}
    \text{HD}(\alpha^v_{l,i}) =
    \begin{cases}
        0 &\text{if } \alpha^v_{l,i} \geq 0\\
        1 &\text{otherwise.}
    \end{cases}
    \label{eq:hdd}
\end{equation}
Consequently, the \ac{PM} is a cost function and, thus, the smallest \acp{PM} 
values belong to the most probable paths which survive the sorting step. The
most probable path is selected as the output of the decoder after the last bit 
decision.

\subsubsection{Fast Simplified SCL Decoding}\hfill\break
To simplify \ac{SCL} decoding and reduce its latency, \ac{FSSCL} 
decoding \cite{hashemi2017flexfastsscl} prunes the \ac{PFT} at nodes
representing constituent codes, which can be decoded directly:
\begin{itemize}
    \item Rate-0 nodes (all bits are frozen, \eg node $3$ in \autoref{subfig:polarTree})
    \item \Ac{REP} nodes (all bits are frozen except the right-most one, \eg node $5$ in \autoref{subfig:polarTree})
    \item \Ac{SPC} nodes (all bits are information except the left-most one, \eg node $6$ in \autoref{subfig:polarTree}) 
    \item Rate-1 nodes (all bits are information, \eg node $14$ in \autoref{subfig:polarTree}).
\end{itemize}
The traversal of corresponding subtrees is not needed. Pruning these subtrees,
\ie omitting the respective computations, reduces the latency and the 
computational complexity.

To differentiate between
a \textit{global} and \textit{node-related} notation, we use $\mathcal{I}_v$ to 
denote the set of information bits in node $v$. The number of information 
bits within node $v$ is $|\mathcal{I}_v| \leq N_v$.
The bijection from the node-related indexing $i_v \in \mathbb Z_{N_v}$ of a
node~${v\in \mathbb Z_{2^{n-s}-1:2^{n-s+1}-1}}$ to the bit index 
$i \in \mathbb Z_N$ is given by
\begin{equation}
    i(i_v) = i_v + \underbrace{
        \overbrace{
            2^s \cdot v + \sum_{j=0}^{s-1}\left(2^j\right)
        }^{\substack{\text{left-most leaf node}\\\text{of node $v$}}}
        - N - 1}_\text{bit index of a leaf node}.
    \label{eq:node2bitIdx}
\end{equation}

The generation of candidate paths in the optimized nodes, \ie proper path
splitting, was addressed in numerous works, since it is usually infeasible and 
redundant to consider all ${L\cdot 2^{|\mathcal{I}_v|}}$ possible candidate 
paths. It was shown in \cite{hashemi2017flexfastsscl} that only 
$P_\text{Rate-1} = \min\LB L-1, N_v\RB$ and $P_\text{SPC} = \min\LB L, N_v\RB$
path splits are needed in Rate-1 and \ac{SPC} nodes of size $N_v$, respectively. Based on the inherent partial 
order of the candidates, their number can be further reduced within the list of 
candidate paths, as shown in \cite{johannsen2022rate1,johannsen2022spc} for 
both node types, respectively.

The \ac{PM} update of candidate path $p$ originating from input path $l$ is carried out
according to Eq.~(\ref{eq:pm}), which is restated as node-related formulation with 
the \acp{PM} of the input
paths $\text{PM}^v_l$
\begin{equation}
	\text{PM}^v_{p} =
	\text{PM}^v_{l} + \sum_{\substack{i_v=0\\
			\beta^v_{p,i_v}\neq \text{HD}\LB\alpha^v_{l,i_v}\RB}}^{N_v-1}
	\left|\alpha^v_{l,i_v} \right|.
	\label{eq:pm_new}
\end{equation}

\subsubsection{SCL Decoding with Dynamic Frozen Bits}\hfill\break
\acp{PTPC} can be efficiently decoded by \ac{SCL} decoding
\cite{DynamicFrozenBits, ParityCheckPolarCode}. The values of the dynamic
frozen bits are calculated according to Eq.~(\ref{eq:dynamicFrozenBit}) with the 
previously decoded $\hat{\boldsymbol{u}}_{0:d}$, whenever a bit with index
$d \in \mathcal{D}$ is decoded. For row-merged polar codes, according to 
\autoref{subsubsec:rowmerged}, $\hat{u}_d = \hat{u}_r$ of belonging 
decoding paths with $\LB r,d\RB \in \mathcal R$ and the \ac{PM} update
follows Eq.~(\ref{eq:pm}).

For \ac{FSSCL}, only the decoding of \ac{SPC} nodes with dynamic frozen was 
addressed so far in \cite{khebbou2023spcdynfrozen} with more details provided 
in \autoref{subsec:SPCDynFrozen}.

\subsection{\ac{FSSCL} Decoding with Dynamic Frozen Bits}
\label{subsec:fastSSCLdf}
For the decoding of \acp{PTPC}, the values of specified information bits
according to the pre-transformation matrix $\boldsymbol{T}$ are needed to set
the corresponding dynamic frozen bits. In the following, we restrict ourselves
to the description of row-merged polar code decoding. However, the described 
steps can be generalized to \acp{PTPC} by replacing the single repeated
information bit index $r$, $(r,d) \in \mathcal{R}$, by all relevant indices 
according to $\boldsymbol{T}$ and the corresponding dynamic frozen bit value
given by Eq.~(\ref{eq:dynamicFrozenBit}).

\subsubsection{Repeated Information and Dynamic Frozen Bits}
\label{subsec:repInfoDynFrozen}\hfill\break
In \ac{FSSCL} with a pruned \ac{PFT}, the bit decisions are made in the nodes 
$v$ of layers $s>0$, resulting in candidate codewords (partial sums) of the corresponding
constituent codes. Thus, the values of the information bits $\hat{u}_i$,
$i \in \mathcal I$ are not directly known. However, the repeated information 
bits $\hat{u}_r$ need to be stored for the decoding of the corresponding 
dynamic frozen bits~$\hat{u}_d$ with ${d \in \mathcal D}$ and
$(r,d) \in \mathcal R$. 
The example of \autoref{fig:polarGraphTree} illustrates two row-merges 
$\mathcal{R}=\{(5,10), (6,12)\}$, where $\hat{u}_6$ and $\hat{u}_5$ are decoded
in the \ac{SPC} node $4$ and set the dynamic frozen bits $\hat{u}_{10}$ and
$\hat{u}_{12}$ in the \ac{REP} node $5$ and the \ac{SPC} node $6$, respectively.
The relevant information bits are extracted from the partial sums
$\boldsymbol{\beta}^v_l$ after applying the polar transform as
\begin{equation}
    \hat{u}_{l,r} = \left(\boldsymbol \beta_l \cdot
        \boldsymbol G_{N_V}\right)_{i_v} \,\Big\vert\,
        i_v\in\mathbb Z_{N_v}:i(i_v)=r.
    \label{eq:repeatedinformation}
\end{equation}
This step is called \ac{IBE}.

In general, an optimized node $v$ of size $N_v=2^s$ can contain different
numbers of frozen bits. The (row-merge) pre-transformation leads to different
constellations of dynamic and static frozen bit positions within these nodes.
To consider the repeated information in the dynamic frozen bits, we introduce
the dynamic frozen vector $\boldsymbol{\delta}^v_l$ of length $N_V$, dedicated
to each of the $L$ decoding paths. The bit values in $\boldsymbol{\delta}$ are
set according to
\begin{equation}
    \delta^v_{l,i_v} = 
    \begin{cases}
        \hat u_{o_r(l),r} & \text{if }
                            i(i_v)=d \in \mathcal D,\,
                             (r,d) \in \mathcal R\\
        0                 & \text{otherwise},
    \end{cases}
    \label{eq:dynamicfrozenvector}
\end{equation}
with $o_r(l)\in\mathbb Z_{L}$ being the origin path index at the info
bit index $r$ of the current dynamic frozen bit index $d$ and of the current 
path $l$. The recovery of the dynamic frozen vector $\boldsymbol \delta^v_l$
is denoted by \ac{DR}.

To transform the dynamic frozen vector $\boldsymbol{\delta}^v_l$ to the partial
sum level of node $v$, it is encoded according to Eq.~(\ref{eq:polarEncoding}) for
all $l \in \mathbb Z_L$, resulting in the dynamic frozen partial sum 
\begin{equation}
    \boldsymbol \chi^v_l = \boldsymbol \delta^v_l \cdot \boldsymbol G_{N_V}.
    \label{eq:dynamicfrozenpartialsum}
\end{equation}

\subsubsection{Rate-0 Nodes}
\label{subsec:Rate0DynFrozen}\hfill\break
For Rate-0 nodes, the partial sums $\boldsymbol \beta^v_l$ are simply set to
$\boldsymbol \chi^v_l$, corresponding to the case of static frozen bits where
they are set to the all-zero vector of size $N_v$:
\begin{equation}
    \boldsymbol \beta^v_l = \boldsymbol \chi^v_l.
    \label{eq:rate0DynFrozen}
\end{equation}
The \acp{PM} are updated according to Eq.~(\ref{eq:pm_new}).

\subsubsection{Repetition Nodes}
\label{subsec:REPDynFrozen}\hfill\break
The decoding of \ac{REP} nodes with dynamic frozen bits follows the rules of the
static frozen bit case with one additional constraint: Similar to the update 
rule in the $g$-function, where the bit values $\boldsymbol{\beta}^l$ from the
left child node determine the sign of the first \ac{LLR}-input, in \ac{REP}
nodes, the values of the (dynamic) frozen bits determine, whether an addition or
a subtraction of the corresponding \ac{LLR} value needs to be carried out to
calculate the decision \ac{LLR} of the information bit. Furthermore, the bit
values in the encoded dynamic frozen vector $\boldsymbol{\chi}^v_l$ 
\textit{XOR} the bit decision for the information bit determine the values of 
the partial sums:
\begin{equation}
    \beta^v_{l,i} = 
        \text{HD}\left(\sum\limits_{j=0}^{N_V-1}
            \left(\alpha^v_{l,j} \cdot 
            \left(-1\right)^{\chi^v_{l,j}}
        \right)\right) \oplus \chi^v_{l,i}.
    \label{eq:repDynFrozen}
\end{equation}
To split the decoding paths, all bits in the partial sum vectors are flipped as
\begin{equation}
    \beta^v_{l+L,i} = 
        \beta^v_{l,i} \oplus 1.
    \label{eq:repCandDynFrozen}
\end{equation}
Again, the \acp{PM} for all resulting $2L$ candidate paths are updated
according to Eq.~(\ref{eq:pm_new}).

\subsubsection{Single Parity Check Nodes}
\label{subsec:SPCDynFrozen}\hfill\break
The decoding of \ac{SPC} nodes with a static frozen bit $\hat u_f=0$, 
$f\in \mathcal{F}$ is based on the fulfillment of the even parity constraint 
induced by $\hat u_f$. Decoding \ac{SPC} nodes with a dynamic frozen bit is 
addressed in \cite{khebbou2023spcdynfrozen} and can be summarized as follows:
The value of the (dynamic) frozen bit determines, whether an even ($\hat u_d=0$) 
or an odd ($\hat u_d=1$) parity constraint needs to be fulfilled. This can be 
achieved by incorporating $\hat u_d$ in the parity calculation of the input 
paths as
\begin{equation}
    \gamma^v_l = \bigoplus_{i=0}^{N_v-1}\left(
                \text{HD}\left(\alpha^v_{l,i}\right)\right) 
                \oplus \hat u_{o_r(l),d}.
    \label{eq:parity_dynfrozen}
\end{equation}
The rest of the decoding of the \ac{SPC} nodes proceeds as in the static frozen 
bit case.

\begin{figure*}
    \centering
	\resizebox{\columnwidth}{!}{\input{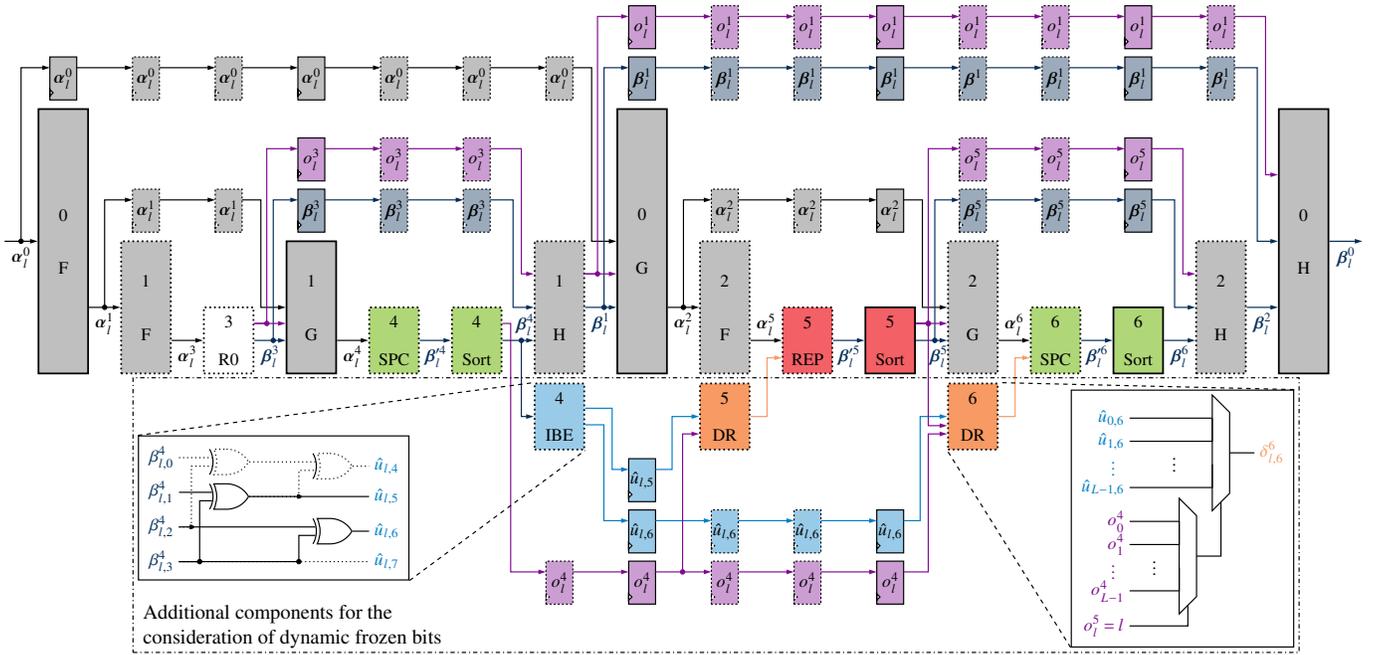}}
    \caption{\footnotesize 
        Simplified representation of a fully unrolled and fully pipelined
        \ac{FSSCL} decoder architecture for a $\mathcal C(16,7)$ with 
        $\mathcal{R}=\{(5,10), (6,12)\}$. Solid shapes represent stages with
        output registers and registers which are set in the delay lines. Stages
        with dotted shapes are combinatorial and the corresponding delay line
        registers are omitted.
        The coloring of the stages (kernels implementing the decoding functions)
        corresponds to \autoref{subfig:polarTree}, registers and signals for
        \acp{LLR} ($\boldsymbol{\alpha}$), partial sums~($\boldsymbol{\beta}$),
        merged information bits ($\hat{u}$) and path pointers ($o$) are colored
        in gray, dark blue, light blue and purple, respectively.
    }
    \label{fig:scl_architecture}
\end{figure*}

\section{Decoder Architecture and Hardware Implementation}
\label{sec:hardware}

This section presents a fully pipelined and unrolled architecture template for 
\ac{FSSCL} decoders with dynamic frozen bits. The architecture is based on
building blocks which implement computational kernels corresponding to the 
decoding functions executed during
the traversal of the \ac{PFT}. Again, the description is restricted to
row-merged polar codes, but can be generalized to \acp{PTPC}.

A simplified representation of an example \ac{FSSCL} decoder architecture for a
$\mathcal C(16,7)$ with $\mathcal{R}=\{(5,10), (6,12)\}$, which corresponds to
the example of \autoref{fig:polarGraphTree}, is shown in
\autoref{fig:scl_architecture}. The given example architecture is simplified with respect
to the omitted representation of the dimension of $L$ and some signals like the
clock and the \acp{PM}.
The delay lines are 
represented by shift registers. However, dependent on the memory depth, more
efficient memory blocks implementing circular buffers based on clock-gated
registers are used in implemented decoders. 
Furthermore, \autoref{fig:scl_architecture} shows a deeply pipelined 
architecture, \ie every stage retains the output signals in registers, 
resulting in a huge overall memory requirement. To reduce the memory
footprint, it was shown in~\cite{kestel2020unrolledList}, that fully pipelined
architectures with balanced register insertion in the pipeline and the delay 
line are more efficient than partially pipelined architectures, eliminating 
only registers in the delay line. Register balancing is based on a timing 
characterization of the building blocks and an automated timing engine, 
which inserts registers to match the delay constraints for a given target 
frequency. This register balancing is also indicated in the example architecture
in \autoref{fig:scl_architecture}: Dotted borders represent combinatorial
stages (without output register) where the corresponding registers in the delay
lines are omitted.

To build a decoder for a \ac{PTPC}, 
according to \autoref{subsec:fastSSCLdf}, 
some of the computational kernels of a plain polar 
code decoder need to be adapted and additional ones are needed:

The additional kernels correspond to the \ac{IBE} and the \ac{DR}, as described
in \autoref{subsec:repInfoDynFrozen}. Furthermore, corresponding delay lines of
the needed bit and path pointer signals are needed. The building blocks of the 
additional kernels and the additional delay lines are shown under the main 
pipeline of the decoder in the lower part of \autoref{fig:scl_architecture}. 

The \ac{IBE} kernel implements 
Eq.~(\ref{eq:repeatedinformation}) and is basically the polar transform of the
input partial sum vectors. 
The corresponding building block is highlighted in light blue.
The detail view of this block in the bottom left
corner of \autoref{fig:scl_architecture} shows the corresponding circuit with 
the needed \textit{XOR} gates for the polar transform to extract $\hat u_{l,5}$ and 
$\hat u_{l,6}$. Note, that only the \textit{XOR} gates connected to a used
output are needed, \ie the sub-circuit depends on $\mathcal{R}$. More formally, 
for all stages of leafs containing the bit indices $r$ with 
$(r,d) \in \mathcal R$, an \ac{IBE} kernel is instantiated to push 
$\hat{\boldsymbol u}_r$ into an information bit delay line.

The \ac{DR} kernels map the repeated information bits of
the origin paths to the bit positions in the dynamic frozen vector $\boldsymbol \delta^v_l$ of path $l$
as described in Eq.~(\ref{eq:dynamicfrozenvector}). Thus, the path pointers of all
intermediate leaf nodes with sorters between the leafs corresponding to bit
indices $r$ and $d$, $(r,d) \in \mathcal R$, are needed for every dynamic frozen
bit in node~$v$, indexed by $d=i(i_v)$ with $i_v \in \mathbb Z_{N_v}$. 
The corresponding building blocks are colored in orange.
The detailed circuit view in the bottom right corner of
\autoref{fig:scl_architecture} indicates the cascades of multiplexers needed
to back-track the paths for one dynamic frozen bit. For every sorter stage from bit index $r$ to $d$, an
additional input port and multiplexer is needed in the \ac{DR} block 
corresponding to the node considering $\hat u_d$. 

The modified kernels correspond to the optimized leaf nodes with 
dynamic frozen bits, as described in \autoref{subsec:Rate0DynFrozen} to 
\autoref{subsec:SPCDynFrozen}. These kernels have an additional input for the 
dynamic frozen vectors $\boldsymbol \delta^v_l$, which are internally encoded
by the polar transform to observe the dynamic frozen partial sums
$\boldsymbol \chi^v_l$ according to Eq.~(\ref{eq:dynamicfrozenpartialsum}) in 
Rate-0 and \ac{REP} node kernels. $\boldsymbol \chi^v_l$ is then used in the
kernels according to Eq.~(\ref{eq:rate0DynFrozen}), Eq.~(\ref{eq:repDynFrozen}) and
Eq.~(\ref{eq:repCandDynFrozen}).
In the \ac{SPC} nodes, the calculation of $\boldsymbol \chi^v_l$ is unneeded
since the one $\hat{u}_{o_r(l),d} = \delta^v_{l,i_v}$ XOR $0$ stays unchanged 
and, according to Eq.~(\ref{eq:parity_dynfrozen}), the dynamic frozen bit value is
used directly in the parity calculation. As the rest of the node can be
processed as in the static frozen bit case, the parallel \ac{SPC} node
architecture presented in \cite{johannsen2022spc} can be adopted for that part
without changes.
\section{Results}
\label{sec:results}

\subsection{Partial Weight Spectra}

\begin{table}[htp]
    \centering
    \footnotesize
    \caption{\footnotesize Comparison of the achieved $A_{\dmin}$ for $\mathcal{RM}(r,n)$ rate-profile using our proposed row-merging algorithm (Alg.~\ref{alg:row_merge}) and the heuristic of \cite{GelincikRowMerged}.}
    \setlength{\tabcolsep}{3.5pt}
    \resizebox{\columnwidth}{!}{\begin{tabular}{crrcrrcrrcrr}
    \toprule
    & \multicolumn{2}{c}{$r=2$} && \multicolumn{2}{c}{$r=3$} && \multicolumn{2}{c}{$r=4$} && \multicolumn{2}{c}{$r=5$}\\
    \cmidrule{2-3} \cmidrule{5-6} \cmidrule{8-9} \cmidrule{11-12}
    n & Alg.~\ref{alg:row_merge} & \cite{GelincikRowMerged} && Alg.~\ref{alg:row_merge} & \cite{GelincikRowMerged} && Alg.~\ref{alg:row_merge} & \cite{GelincikRowMerged} && Alg.~\ref{alg:row_merge} & \cite{GelincikRowMerged}\\
    \midrule
    5 & 364 & 364 \\ 
    6 & 396 & 580 && 2328 & 3160\\
    7 & 396 & 660 && 2392 & 2884 && 15344 & 20528\\
    8 & 380 & 970 && 2328 & 3104 && 15664 & 40250 && 107664 & 191868\\
    9 & 396 & 1271 && 2328 & 3530 && 15360 & 40400 && 117122 & 1540925\\ 
    \bottomrule
\end{tabular}

}
    \setlength{\tabcolsep}{6pt}
    \label{tab:heuristic}
\end{table}

Using the proposed row-merging algorithm (Alg.~\ref{alg:row_merge}), we design various pre-transformed polar codes.
First, we consider the case of $\mathcal{RM}(r,n)$ rate-profiles, and compare the results to the design proposed in \cite{GelincikRowMerged}, for which the optimal value of the hyperparameter $l^*$ has been determined by exhaustive search.
Tab.~\ref{tab:heuristic} lists the achieved number of minimum-weight codewords $A_{\dmin}$ which have been computed using \cite[Alg.~1]{ZunkerTreeIntersection} for different values of $n$ and $r$.
Note that, following Thm.~\ref{thm:dmin}, the codes have the same minimum distance as the plain \ac{RM} codes, i.e., $d_\mathrm{min}=2^{n-r}$.
For all parameter combinations, our proposed selection of row-merges results in codes with a lower (or equal) value of $A_{\dmin}$ and hence, better (or the same) \ac{ML} performance compared to \cite{GelincikRowMerged}. 
In the case $n=9$, we restricted our proposed algorithm to evaluate only the two subsequent yet unused frozen rows for each $\dmin$-weight generator to reduce computational complexity, as discussed in Rem.~\ref{rk:simplification}.
Despite this simplification, the achieved difference to \cite{GelincikRowMerged} are particularly large for $n=9$. %

\begin{table}[htp]
    \centering
    \footnotesize
    \caption{\footnotesize Row-merges for the proposed \acp{PTPC}.}
    \setlength{\tabcolsep}{1pt}
    \resizebox{\columnwidth}{!}{\begin{tabular}{cl}
    \toprule
    Rate-profile & Row-merges\\
    \midrule
    \multirow{3}{*}{DE(128,60)@\SI{2.9}{\dB}} & \scriptsize$\{(29,34),(30,35),(43,70),(45,50),(46,73),(51,68)$\\
    & \scriptsize$\phantom{\{}(53,74),(54,69),(57,66),(58,67),(60,65),(75,100),$\\
    & \scriptsize$\phantom{\{}(78,81),(83,104),(85,98),(86,112),(92,97)\}$\\
    \midrule
    \multirow{5}{*}{DE(256,75)@\SI{0.1}{\dB}} & \scriptsize$\{(115,133),(117,134),(118,129),(121,131),(122,135),$\\
    & \scriptsize$\phantom{\{}(124,130),(157,162),(158,163),(167,201),(171,198),$\\
    & \scriptsize$\phantom{\{}(173,178),(174,208),(179,197),(181,194),(182,202),$\\
    & \scriptsize$\phantom{\{}(185,204),(186,195),(188,193),(199,232),(206,209),$\\
    & \scriptsize$\phantom{\{}(211,240),(213,226),(217,228),(218,225)\}$\\
    \midrule
    \multirow{3}{*}{DE(1024,512+7)@\SI{0}{\dB}} & \scriptsize$\{(720,769),(720,774),(736,770),(736,777),(808,833),$\\
    & \scriptsize$\phantom{\{}(808,912),(816,960),(834,904),(836,928),(840,900),$\\
    & \scriptsize$\phantom{\{}(848,898),(864,897)\}$\\
    \bottomrule
\end{tabular}
}
    \setlength{\tabcolsep}{6pt}
    \label{tab:row merges}
\end{table}

Next, we design row-merged polar codes taking also the performance under \ac{SCL} into account.
Hence, first the suitable rate-profiles are designed according to Sec. \ref{ssec:rateproofiledesign} and then, Alg.~\ref{alg:row_merge} is applied.
We compare our results to \ac{PAC} codes and 5G \ac{CRC}-aided polar codes.
The partial distance spectra of the codes are computed using \cite{Miloslavskaya2022partialweightdistribution} (5G+CRC) and \cite{PartialEnumPAC} (\ac{PAC} codes).
For row-merged polar codes, we generalized the method from \cite{PartialEnumPAC} to arbitrary \acp{PTPC}. In fact, we apply the same recursive exploration of polar cosets and use \cite{PartialEnumPAC} to find which sub-trees have to be explored to count the codewords of a certain weight.

\begin{table}[htp]
    \centering
    \footnotesize
    \caption{\footnotesize Partial weight spectra of $\mathcal C(128,60)$ polar codes.}
    \resizebox{\columnwidth}{!}{\begin{tabular}{ccrrrrr}
    \toprule
    Rate-profile & Pre-transform & $A_8$ & $A_{12}$ & $A_{16}$ & $A_{18}$\\
    \midrule
    \multirow{2}{*}{5G(128,60+11)} & Identity & 452 & 5194 & 242404 & 0\\ 
    & CRC-11 & 0 & 6 & 580 & 0\\
    \midrule
    \multirow{4}{*}{DE(128,60)@\qty{2.9}{\dB}} & Identity & 0 & 0 & 28952 & 0\\
    & Convolution & 0 & 0 & 2136 & 280\\
    & Rep. (Alg.~\ref{alg:row_merge}) & 0 & 0 & 2328 & 1114\\
    & Rep. (simpl.) & 0 & 0 & 2328 & 1284\\
    \midrule
    $\mathcal I_\mathrm{min}=\{27\}$ & Identity & 0 & 0 & 33048 & 0\\
    \bottomrule
\end{tabular}}
    \label{tab:weightspectra_128_60}
\end{table}

For the short block length regime, we design a $(128,60)$ code to be comparable to existing literature \cite{Kestel2023,SCAL}. 
However, unlike the code from \cite{Kestel2023,SCAL} with $\mathcal{I}_\mathrm{min}= \{ 27 \}$ which is designed to have a large automorphism group, the code design for row-merged codes is not required to be symmetric. As a rate-profile, we take the design from Ex.~\ref{ex:P60128} using \ac{DE} with design \ac{SNR} $\qty{2.9}{dB}$.
The set of row-merges $\mathcal{R}$ is obtained using Alg.~\ref{alg:row_merge} and listed in Tab.~\ref{tab:row merges}. 
For reference, we also exhaustively searched for the optimal \ac{PAC} code (with $\operatorname{deg}(p(x))\leq20$) convolutional precoding polynomial $1047_8$.
Tab.~\ref{tab:weightspectra_128_60} gives the partial weight spectra of these codes. 
The 5G polar code with 11-bit \ac{CRC} is also provided for reference. 
Note that for this code, the last 11 bits are frozen, i.e. the rate-profile is not decreasing. In fact, all minimum-weight cosets are pre-transformable and, hence, $\dmin$ can be increased. However, we see that the \ac{CRC} cannot fully remove all weight-12 codewords.
The \ac{DE}-based design and $\mathcal{I}_\mathrm{min}= \{27\}$ already have a minimum distance of $16$.
Both row-merge and convolutional precoding perform similarly in expurgating the base code and leave $2328$ and $2136$ minimum-weight codewords, respectively. 
Thus, they are expected to perform very similarly in terms of error-rate. Note that the simplified search from Rem.~\ref{rk:simplification} results in the same $\Admin$.

\begin{table}[htp]
    \centering
    \footnotesize
    \caption{\footnotesize Partial weight spectra of $\mathcal C(256,75)$ polar codes.}
    \resizebox{\columnwidth}{!}{\begin{tabular}{ccrrrrr}
    \toprule
    Rate-profile & Pre-transform & $A_{16}$ & $A_{24}$ & $A_{32}$ & $A_{34}$\\
    \midrule
    \multirow{2}{*}{5G(256,75+11)} & Identity & 452 & 5194 & 275656 & 0\\ 
    & CRC-11 & 0 & 6 & 664 & 0\\
    \midrule
    \multirow{4}{*}{DE(256,75)@\qty{0.1}{\dB}} & Identity & 0 & 0 & 46104 & 0\\
    & Convolution & 0 & 0 & 2136 & 32\\
    & Rep. (Alg.\ref{alg:row_merge}) & 0 & 0 & 2328 & 352\\
    & Rep. (simpl.) & 0 & 0 & 2328 & 352\\
    \bottomrule
\end{tabular} %
}
    \label{tab:weightspectra_256_75}
\end{table}

For a medium block length, we pick the parameters $N=256$ and $K=75$ as in Ex.~\ref{ex:P75256}.
Our \ac{PTPC} design uses a rate-profile form \ac{DE} at $\qty{0.1}{dB}$ with minimum information set $\mathcal I_\mathrm{min}=\{63,115,157,167\}$ and $d_\mathrm{min}=32$. 
Repetition and convolutional precoder with $\mathcal{R}$ as listed in Tab.~\ref{tab:row merges} and polynomial $2213_8$ (optimal for $\operatorname{deg}(p(x))\leq20$), respectively, have been constructed.
Tab.~\ref{tab:weightspectra_256_75} lists the partial weight spectra of the codes. 
Here, the simplification from Rem.~\ref{rk:simplification} results in the same $\Admin$ and $A_{34}$, but not the same code.

Using convolutional and row-merge precoders, the number of minimum-weight codewords is drastically reduced, while the 11-bit \ac{CRC} cannot fully eliminate all weight-$24$ codewords of the 5G polar code with originally $\dmin=16$. Still, the number of low-weight codewords (below weight 34) is lower than the proposed codes. Therefore, we expect the \ac{ML} bound of the 5G code to be slightly better in the low \ac{SNR} regime.

\begin{table}[htp]
    \centering
    \footnotesize
    \caption{\footnotesize Partial weight spectra of $\mathcal C(1024,512)$ polar codes.}
    \setlength{\tabcolsep}{5pt}
    \resizebox{\columnwidth}{!}{\begin{tabular}{ccrrrrr}
    \toprule
    Rate-profile & Pre-transform & $A_{16}$ & $A_{24}$ & $A_{28}$\\
    \midrule
    \multirow{2}{*}{5G(1024,512+11)} & Identity & 42012 & 9745930 & 10305472 \\ 
    & CRC-11 & 57 & 5655 & 7193\\
    \midrule
    \multirow{3}{*}{DE(1024,512+7)@\qty{0}{\dB}} & Identity & 3136 & 49664 & 16384\\ 
    & Rep. (Alg.~\ref{alg:row_merge}) & 0 & 11008 & 128\\
    & Rep. (simpl.) & 0 & 11532 & 400\\
    \bottomrule
\end{tabular}
}
    \label{tab:weightspectra_1024_512}
    \setlength{\tabcolsep}{6pt}
\end{table}

Finally, a rate-$\nicefrac{1}{2}$ code with $N=1024$ and $K=512$ is analyzed. For such long lengths, \ac{SC} decodability is the key criterion to obtain a good performance and we make use of the second proposed method, i.e., a non-decreasing rate profile. We use the base code from Ex.~\ref{ex:P5121024} obtained from \ac{DE} with design \ac{SNR} of $\qty{0}{dB}$ with $K'=512+7$ and freeze the last $7$ minimum-weight generators.
The weight spectrum of this code is given in Tab.~\ref{tab:weightspectra_1024_512}. 
Using the proposed Alg.~\ref{alg:rate_profile_design}, we can find row merges that eliminate all weight-16 and weight-20 codewords to achieve a minimum distance of 24. The row merges are listed in Tab.~\ref{tab:row merges}.
In contrast, the 5G code with the same parameters and \ac{CRC}-11 still contains some weight-16 codewords.
In this case, the simplification from Rem.~\ref{rk:simplification} would result in a row-merged code with $\dmin=24$ and $\Admin=11532$.

\subsection{Design Complexity}
\begin{table}[htp]
    \centering
    \footnotesize
    \caption{\footnotesize Running time required to select the row-merges using Alg.~\ref{alg:row_merge} and its simplification from Rem.~\ref{rk:simplification}.}
    \setlength{\tabcolsep}{4pt}
\begin{tabular}{cccc}
    \toprule
    Code & $\mathcal{C}(128,60)$ & $\mathcal{C}(256,75)$ & $\mathcal{C}(1024,512)$\\
    \midrule
    Alg.~\ref{alg:row_merge} & $\SI{1.83}{s}$ & $\SI{9.84}{s}$ & $\SI{191}{s}$ \\
    simplified (Rem.~\ref{rk:simplification}) & $\SI{582}{ms}$ & $\SI{1.57}{s}$ & $\SI{48.5}{s}$ \\
    \bottomrule 
\end{tabular}
\setlength{\tabcolsep}{6pt}
    \label{tab:designtime}
\end{table}
In Tab.~\ref{tab:designtime}, we list the required design time, i.e., the time to select the row-merge pairs for each of the three codes and compare Alg.~\ref{alg:row_merge} to its simplification from Rem.~\ref{rk:simplification}.
We use a single-core C implementation of \cite[Alg.~1]{ZunkerTreeIntersection} for enumerating the $\wmin$-weight codewords.
For the enumeration of codewords with higher weight, a non-parallel Numba-accelerated Python implementation of the method in \cite{PartialEnumPAC} is used.
As we can see, the restriction of the search speeds up the design by roughly a factor of $3$ to $6$. Furthermore, the use of the general enumeration algorithm of \cite{PartialEnumPAC} considerably increases the design time of the $\mathcal{C}(1024,512)$ compared to the smaller codes.

\subsection{Pre-Transform Complexity}
\begin{table}[htp]
    \centering
    \footnotesize
    \caption{\footnotesize Number of dynamic frozen bits $|\mathcal{D}|$ and required XOR operations $N_\mathrm{XOR}$ of the pre-transform. }
    \setlength{\tabcolsep}{4pt}
\begin{tabular}{crrcrrcrr}
    \toprule
    Code & \multicolumn{2}{c}{$\mathcal{C}(128,60)$} && \multicolumn{2}{c}{$\mathcal{C}(256,75)$} && \multicolumn{2}{c}{$\mathcal{C}(1024,512)$}\\
    \cmidrule{2-3}\cmidrule{5-6}\cmidrule{8-9}
    Pre-transform & $|\mathcal{D}|$ & $N_\mathrm{XOR}$ && $|\mathcal{D}|$ & $N_\mathrm{XOR}$ && $|\mathcal{D}|$ & $N_\mathrm{XOR}$ \\
    \midrule
    CRC-11 & 11 & 375 && 11 & 473 && 11 & 2960\\
    Convolution & 35 & 142 && 59 & 257 && --- & ---\\
    Repetition & 17 & 0 && 24 & 0 && 13 & 0 \\
    \bottomrule 
\end{tabular}
\setlength{\tabcolsep}{6pt}
    \label{tab:precoding checks}
\end{table}
Next, we compare the different pre-transforms in their computational complexity for encoding and decoding. 
While it is generally difficult to quantify the overhead introduced by the pre-transformation exactly, simplified metrics are the number of dynamic frozen bits $|\mathcal{D}|$ and the number of XOR operations required to compute each one according to Eq.~(\ref{eq:dynamicFrozenBit}).
Tab.~\ref{tab:precoding checks} lists these metrics for the previously introduced codes and precoders. Note that we assume the pre-transforms to be systematic in all cases.
As we can see, \ac{CRC} has a low number of dynamic checks, however, with many operations needed to compute them. 
\ac{PAC} codes on the other hand have a very high number of dynamic frozen bits (virtually all frozen bits after the first information bit are dynamic), and also many XOR operations are required. 
In contrast, our proposed row-merged polar codes have a comparably small number of dynamic frozen bits and do not need any computation overhead.

\subsection{Error-Correction Performance}\label{subsec:errorCorrection}
We evaluate the error-correction performance of the designed \acp{PTPC} with \ac{BPSK} mapping over an \ac{AWGN} channel using Monte-Carlo simulation with a minimum of 1000 block errors. All simulation results are obtained using floating-point and ``box-plus'' $f$-function implementation. 
For reference, we also show the saddle point approximations of the \ac{PPV} meta converse bound \cite{SaddlePointApproxMC} %
and the random coding union bound \cite{SaddlePointApproxRCU} (implementation based on \cite{MCandRCUimplementation}). 
Moreover, we find lower bounds on the \ac{ML} performance of the codes using by using box-and-match decoding \cite{bma} of order 4 for $N=128$ and \ac{SCL} decoding with growing list size (up to $L=65536$) for the longer codes. If a codeword closer to the received sequence than the transmitted codeword is found, this is recorded as an \ac{ML} error, while in all other cases we assume the \ac{ML} decoder would have succeeded. Hence, this is only a lower bound.
For the low \ac{BLER} regime where this is still infeasible, we obtain truncated union bounds from the partial weight spectra.

\begin{figure}[htp]
	\centering
	\resizebox{\columnwidth}{!}{\input{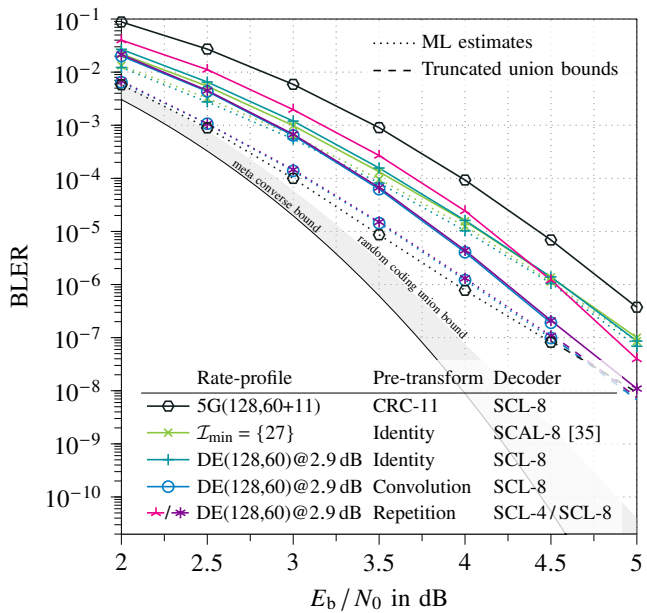}}
	\caption{\footnotesize \ac{BLER} comparison of $\mathcal C(128,60)$ polar codes (CRC-aided, plain, PAC and row-merged).}
    \label{fig:polar_128_bler}
\end{figure}

Fig.~\ref{fig:polar_128_bler} shows the \ac{BLER} performance of the $(128,60)$ codes.
The designed \acp{PTPC} outperform the 5G polar code by \SI{0.6}{dB} at a \ac{BLER} of $10^{-5}$ under \ac{SCL} decoding with $L=8$. As expected from the similar weight spectra, both the \ac{PAC} and the row-merged polar code have almost an identical performance. Compared to the $\mathcal{I}_\mathrm{min}= \{ 27 \}$ code under \ac{SCAL} decoding (with automorphism selection optimized at \SI{4}{dB} \cite{SCAL}), the row merged polar code only requires a list size $L=4$ to reach the same \ac{BLER} of $10^{-6}$. 
Note that the $\mathcal{I}_\mathrm{min}= \{ 27 \}$ code is mainly limited by its poor weight spectrum, as can be seen in the \ac{ML} performance. In fact, the plain polar code obtained using \ac{DE} is already slightly better than the $\mathcal{I}_\mathrm{min}= \{ 27 \}$ code which had to be selected with symmetry in mind.

\begin{figure}[htp]
	\centering
	\resizebox{\columnwidth}{!}{\input{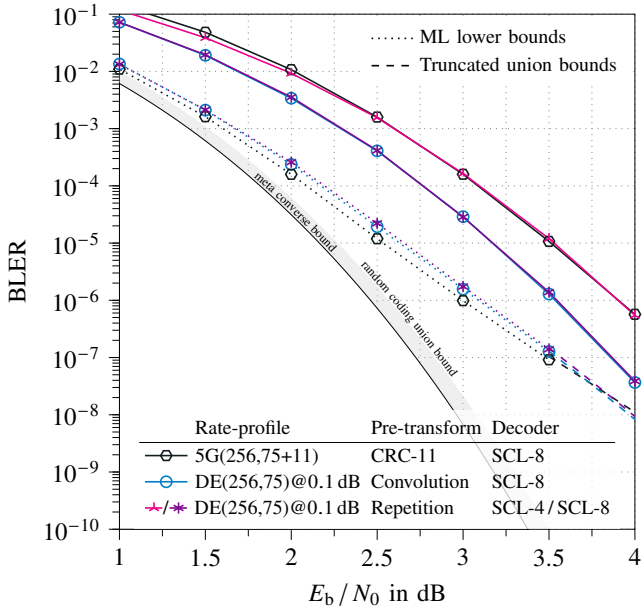}}
	\caption{\footnotesize \ac{BLER} comparison of $\mathcal C(256,75)$ polar codes (CRC-aided, PAC and row-merged).
    }
    \label{fig:polar_256_bler}
\end{figure}

Fig.~\ref{fig:polar_256_bler} shows the \ac{BLER} performance of the $(256,75)$ codes. Using \ac{SCL}-8, the proposed \acp{PTPC} outperform the 5G polar code by \SI{0.3}{dB} at a \ac{BLER} of $10^{-5}$. 
The advantage of the row-merged code is smaller, as the rate loss caused by the CRC-11 is lower relative to the code length than for $N=128$.
Again, repetition and convolutional pre-transform perform almost identically. The row-merged polar code reaches the same performance as the 5G polar code under \ac{SCL}-8 using only a list size of $L=4$.

\begin{figure}[htp]
	\centering
	\resizebox{\columnwidth}{!}{\input{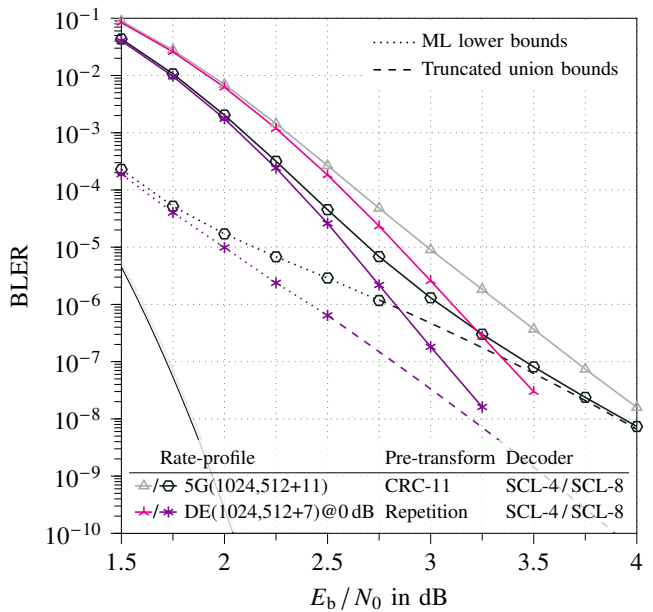}}
	\caption{\footnotesize \ac{BLER} performance of $\mathcal C(1024,512)$ polar codes (CRC-aided vs. row-merged).
    }
    \label{fig:polar_1024_bler}
\end{figure}

Finally, we show the \ac{BLER} performance of the designed $(1024,512)$ code in Fig.~\ref{fig:polar_1024_bler} under \ac{SCL}-4 and \ac{SCL}-8. We remark that the SCL-8 curve is only a lower bound, as it is obtained from the cases where SCL-4 outputs a codeword estimate that is neither correct nor more likely than the transmitted codeword.
The performance of the 5G \ac{CRC} aided polar code is clearly limited by the poor distance spectrum, which can be seen in the shallow truncated union bound.
In contrast, the proposed row-merged polar code exhibits a significantly better distance spectrum and, thus, outperforms the 5G code for all \ac{SNR} values for the same list size. Above $E_\mathrm{b}/N_\mathrm{0} = \SI{3.25}{dB}$, the row-merged code under \ac{SCL}-4 even outperforms \ac{SCL}-8 for the 5G code.

\subsection{Implementation Results}

In this section, we present unrolled and fully pipelined ASIC implementations 
of 8 different decoders for a target clock frequency of $f_\mathrm{clk}=\SI{500}{\mega\hertz}$.
Synthesis and \ac{PAR} were executed with the \textit{Synopsys} tools
\textit{Design Compiler} and \textit{IC-Compiler}, respectively. A \nm{12}
FinFET technology from \textit{GlobalFoundries} is used under worst case 
\ac{PVT} conditions (\SI{125}{\degreeCelsius}, \SI{0.72}{\volt}) for timing and 
nominal case \ac{PVT} (\SI{25}{\degreeCelsius}, \SI{0.8}{\volt}) for power. 
Power values stem from post-\ac{PAR} netlist simulations with back-annotated 
wiring data and test data at the $E_\mathrm{b}/N_0$ needed for a \ac{BLER} 
$=10^{-5}$.

\begin{table*}[htp]
    \centering
    \caption{\footnotesize
        Post-\ac{PAR} ASIC implementation results.
    }
    \label{tab:results_all}
    \footnotesize
    \begin{NiceTabular}{l|ccccc|ccc}
\toprule
    Code
    & \multicolumn{5}{c}{$\mathcal{C}(128,60)$}
    & \multicolumn{3}{c}{$\mathcal{C}(256,75)$}
    \\
    \midrule
    Rate-profile
    & 5G
    & $\mathcal I_{\mathrm{min}}\!=\!\{27\}$
    & \multicolumn{3}{c}{DE@\SI{2.9}{\dB}
    }
    & 5G
    & \multicolumn{2}{c}{DE@\SI{0.1}{\dB}
    } 
    \\
    \cmidrule{4-6}
    Pre-transform
    & CRC-11
    & Identity
    & Identity
    & \multicolumn{2}{c}{Repetition}
    & CRC-11
    & \multicolumn{2}{c}{Repetition} 
    \\
    \cmidrule{5-6}
    \cmidrule{8-9}
    Decoder
    & SCL-8
    & SCAL-8
    & SCL-8
    & SCL-8
    & SCL-4
    & SCL-8
    & SCL-8
    & SCL-4 
    \\
\midrule
Frequency $f_\mathrm{clk}$  / {(}MHz{)}       
    & 500 
    & 500 
    & 500 
    & 500 
    & 500 
    & 500 
    & 500 
    & 500 
\\
Throughput $T_\mathrm{c}$ / {(}Gbps{)}     
    &  64.0 
    &  64.0 
    &  64.0 
    &  64.0 
    &  64.0 
    & 128.0 
    & 128.0 
    & 128.0 
\\
Latency / {(}CC{)}          
    &  29 
    &  32 
    &  32 
    &  32 
    &  24 
    &  51 
    &  50 
    &  34 
\\
Latency / {(}ns{)}          
    &  58.0 
    &  64.0 
    &  64.0 
    &  64.0 
    &  48.0 
    & 102.0 
    & 100.0 
    &  68.0 
\\
Area $A$ / {(}mm$^2${)}         
    & 0.303 
    & 0.381 
    & 0.289 
    & 0.299 
    & 0.115 
    & 0.750 
    & 0.737 
    & 0.288 
\\
\textbf{Area Efficiency $\mu_\mathrm{A}$ / {(}Gbps/mm$^2${)}}
    & \textbf{211.5}
    & \textbf{168.2}
    & \textbf{221.1}
    & \textbf{214.0}
    & \textbf{556.0}
    & \textbf{170.7}
    & \textbf{173.8}
    & \textbf{445.1}
\\
Utilization / {(}\%{)}  
    & 74 
    & 73
    & 73 
    & 74 
    & 73 
    & 73 
    & 73 
    & 75 
\\
Power Total $P$ / {(}W{)}        
    & 0.494
    & 0.575
    & 0.429
    & 0.499
    & 0.190
    & 1.056
    & 1.083
    & 0.541
\\
\textbf{Energy Efficiency $\mu_\mathrm{E}$ / {(}pJ/bit{)}}
    & \textbf{7.73}
    & \textbf{8.99}
    & \textbf{6.71}
    & \textbf{7.79}
    & \textbf{2.97}
    & \textbf{8.25}
    & \textbf{8.46}
    & \textbf{4.23}
\\
Power Density / {(}W/mm$^2${)}
    & 1.63
    & 1.51
    & 1.48
    & 1.67
    & 1.65
    & 1.41
    & 1.47
    & 1.88 
\\
\midrule
$E_\mathrm{b}/N_0$ / (dB) @ BLER $10^{-5}$ 
    & 4.43
    & 4.10
    & 4.10
    & 3.85
    & 4.16
    & 3.52
    & 3.18
    & 3.52
\\
\bottomrule
\end{NiceTabular}

\end{table*}

All decoders use the hardware-friendly min-sum formulations of the $f$-function 
\cite{leroux2011SCdecHW}, a quantization of \SI{6}{\bit} for channel and 
internal \acp{LLR} and \SI{8}{\bit} for \acp{PM}. %
The threshold values for Rate-1 \cite{johannsen2022rate1} and \ac{SPC} nodes 
\cite{johannsen2022spc} were chosen to limit the performance loss to 
$\dB{0.05}$
at a \ac{BLER} of $10^{-5}$ compared to the floating-point results presented in 
\autoref{subsec:errorCorrection}.

The presented values for the coded throughput are given by $T_\mathrm{c}=f_\mathrm{clk}\cdot N$. For the comparisons, we focus on the metrics area efficiency 
$\mu_\mathrm{A}=T_\mathrm{c}/A$ and energy efficiency $\mu_\mathrm{E}=P/T_\mathrm{c}$, with $A$ being the 
area and $P$ the power of the implementation. The power density is a derived 
metric calculated as $P/A=\mu_\mathrm{E}\cdot\mu_\mathrm{A}$. 

\autoref{tab:results_all} shows the post-\ac{PAR} results of decoders for the
selected codes with code lengths $N \in \{128, 256\}$ described and
analyzed in the previous sections, except from the \ac{PAC} codes. We omit the
implementation of the \ac{PAC} code decoders, since \autoref{tab:precoding checks}
shows the higher computational complexity of the convolution pre-transformations
compared to the row-merged codes.

\subsubsection{$N=128$}\hfill\break
Comparing the implementation costs of \ac{SCL}-$8$ decoders for the 5G CRC-11
code and the row-merged polar code, both decoders show almost equal values in 
$\mu_\mathrm{E}$ and $\mu_\mathrm{A}$, although the different codes cause appreciable differences 
in the \acp{PFT} and the corresponding unrolled architectures. While the decoder
for the row-merged code shows a negligible better $\mu_\mathrm{A}$ than the decoder for 
the 5G CRC-11 code, the advantage is vice versa for $\mu_\mathrm{E}$.

In comparison to the \ac{SCAL}-$8$ decoder for the code with 
$\mathcal I_\mathrm{min} = \{27\}$, the row-merged \ac{SCL}-$8$ decoder is 
superior in $\mu_\mathrm{A}$ and $\mu_\mathrm{E}$ by $\times 1.27$ and $\div 1.15$,
respectively. The main reason here is the complete 
initialization of the list in the \ac{SCAL} decoder causing additional area and 
power consumption, which is not needed for \ac{SCL} decoders.

To show the overhead caused by the row-merges, we also implemented the 
\ac{SCL}-$8$ decoder for the plain polar code~({DE@\dB{2.9}}) without the 
row-merges. The observed overhead is 
$\times 1.03$ %
for $\mu_\mathrm{A}$ and 
$\div 1.16$ %
for $\mu_\mathrm{E}$. 
An important contribution to the increased power consumption stems from the 
additional logic needed for the incorporation of the dynamic frozen bits.

Due to the very good error-correction capability of the row-merged polar code,
it is also possible to reduce the list size to $L=4$. The \ac{SCL}-$4$ decoder 
exhibits an improvement of
$\times 2.63$ %
in $\mu_\mathrm{A}$ and 
$\div 2.60$ %
in $\mu_\mathrm{E}$ 
compared to the 5G reference decoder, while still providing a gain of \dB{0.3}
at a \ac{BLER} of $10^{-5}$.

\begin{figure}
    \centering
    \subfloat[\footnotesize SCL-8: \mmsq{0.750}]{%
        \includegraphics[height=0.61\columnwidth]%
        {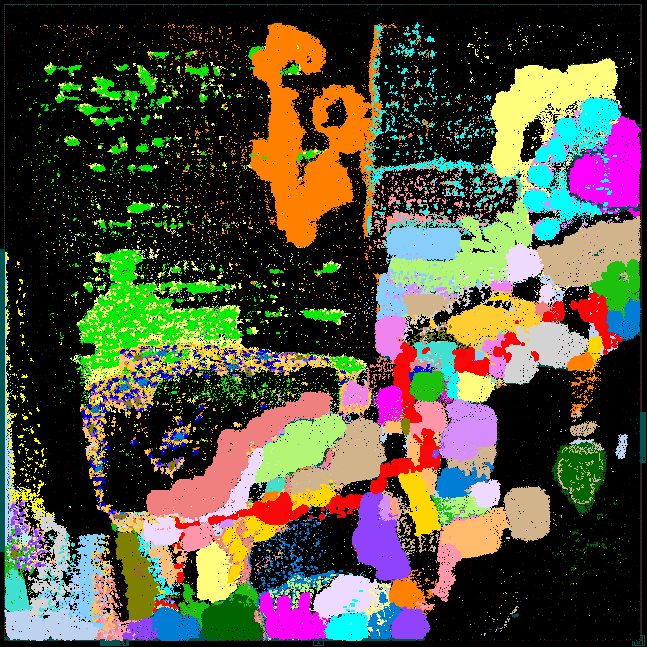}%
        \label{fig:layout_5g_256_scl8}%
    }%
    \hfill
    \centering
    \subfloat[\footnotesize SCL-4: \mmsq{0.288}]{%
        \includegraphics[height=0.38\columnwidth]%
        {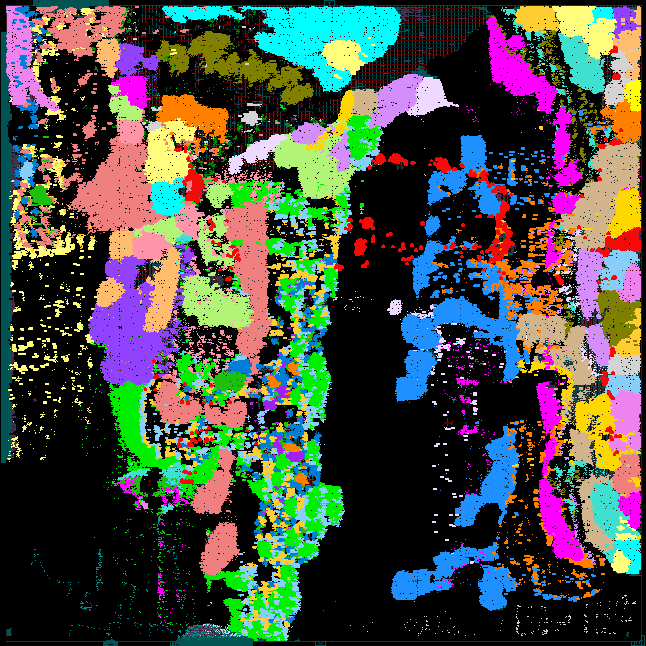}%
        \label{fig:layout_rowmerged_256_scl4}%
    }%
    \caption{\footnotesize
        Layouts with equal area scaling showing the (a) $\mathcal C(256, 75)$ 
        5G CRC-11 code \ac{SCL}-$8$ decoder and (b) the row-merged code
        \ac{SCL}-$4$ decoder providing similar \ac{BLER} performance;
        different colors represent computational kernels, delay line memory is
        colored in black, area portions related to the pre-transformations
        (including memory) is highlighted in red.
    }\label{fig:layouts_256}
\end{figure}

\subsubsection{$N=256$}\hfill\break
For the decoders of the $\mathcal C(256,75)$ codes, the relations in the 
implementation costs are comparable to the ones for the $\mathcal C(128,60)$ 
codes.
Again, the \ac{SCL}-$8$ decoders for the 5G polar code with \ac{CRC}-$11$ and 
the row-merged polar code show similar values for $\mu_\mathrm{A}$ and $\mu_\mathrm{E}$ with a 
gain of \dB{0.3} at a \ac{BLER} of $10^{-5}$.
With a list size reduction to $L=4$, the row-merged code under
\ac{SCL}-$4$ decoding achieves equivalent error-correction performance as the
\ac{SCL}-8 decoder for the 5G CRC-11 code (\autoref{fig:polar_256_bler}),
but with significant improvements in the implementation costs. We observe a 
$\times 2.61$ %
higher $\mu_\mathrm{A}$ and a
$\div 1.95$ %
improvement in $\mu_\mathrm{E}$. 
The smaller area demand of the row-merged \ac{SCL}-$4$ decoder, compared
to the \ac{SCL}-$8$ decoder with CRC-11 providing equal \ac{BLER} performance,
is also visible in the layouts shown in \autoref{fig:layouts_256} with equal 
area scaling. The area occupied by the components related to the 
pre-transformations is highlighted in red. These components are the 
\ac{CRC}-updates and corresponding check-sum memory for the 5G code and the 
\ac{IBE}, \ac{DR} and the corresponding memory for the repeated information 
bits and the path pointers for the row-merged code.

\section{Conclusion and Outlook}\label{sec:conc}
In this paper, we propose methods to construct row-merged polar codes based on the analysis of the formation of minimum-weight codewords of \acp{PTPC}.
As low-complexity weight enumeration algorithms for these codes are now available, our optimized construction algorithms can directly take the actual number of minimum-weight codewords into account in the code design.
The proposed algorithms cover both the construction of the rate-profile and the selection of the row-merges.
Moreover, a corresponding fast simplified \ac{SCL} decoding algorithm is proposed and demonstrated in a fully pipelined hardware implementation.
Numerical results show that the optimized row-merged polar codes outperform \ac{CRC}-aided polar codes in error correcting performance and implementation costs.
The decoding algorithm and its implementation are not limited to row-merged codes but also applicable to other types of \acp{PTPC}.
However, as row-merged polar codes show identical error-correcting performance as \ac{PAC} codes, they should be preferred to \ac{PAC} codes due to their lower complexity overhead.

As the set of row-merges has a direct impact on the circuitry of the decoder, future work includes the optimization of the row-merges to further reduce the hardware complexity, e.g., by favoring groups of dynamic frozen bits in the same leaf node of the \ac{PFT}.

\appendices

\ifCLASSOPTIONcaptionsoff
  \newpage
\fi

\bibliographystyle{IEEEtran}
\bibliography{main}

\end{NoHyper}
\end{document}